\documentclass[11pt, a4paper,reqno]{amsart}
\usepackage{amsfonts}
\usepackage{amssymb}
\usepackage[dvips]{graphics}
\usepackage{epsfig}
\usepackage{euscript}
\usepackage{color}

\usepackage{fullpage,hyperref}

\usepackage{fancyvrb}

 \newtheorem{thm}{Theorem}[section]

 \newtheorem{prop}[thm]{Proposition}
 \theoremstyle{definition}
 
 \newtheorem{rem}[thm]{Remark}

 \numberwithin{equation}{section}
 
 \def\idtyty{{\mathchoice {\mathrm{1\mskip-4mu l}} {\mathrm{1\mskip-4mu l}} %
{\mathrm{1\mskip-4.5mu l}} {\mathrm{1\mskip-5mu l}}}}

\newcommand{\caA}{{\mathcal A}}
\newcommand{\caB}{{\mathcal B}}

\newcommand{\caE}{{\mathcal E}}
\newcommand{\caF}{{\mathcal F}}

\newcommand{\caH}{{\mathcal H}}

\newcommand{\caJ}{{\mathcal J}}

\newcommand{\caL}{{\mathcal L}}

\newcommand{\caN}{{\mathcal N}}

\newcommand{\caT}{{\mathcal T}}
\newcommand{\caU}{{\mathcal U}}
\newcommand{\caV}{{\mathcal V}}

\newcommand{\bbN}{{\mathbb N}}

\newcommand{\bbR}{{\mathbb R}}

\newcommand{\bbT}{{\mathbb T}}

\newcommand{\bbZ}{{\mathbb Z}}

\newcommand{\iu}{\mathrm{i}}

\newcommand{\str}{^{*}}
\newcommand{\ep}[1]{\mathrm{e}^{#1}}

\newcommand{\ket}[1]{\vert#1\rangle}

\DeclareMathOperator*{\wlim}{w-lim}

\newcommand{\cstar}{\mathfrak{A}}

\newcommand{\eigstr}[2]{\mathcal{F}^{#1}_{#2}}
\newcommand{\dstr}[2]{\mathcal{F}^{#1}_{\overline{#2}}}

\newcommand{\plus}[1]{\mathcal{A}^{#1}}
\newcommand{\face}{\mathcal{B}}

\begin{document}
\title{Dynamical Abelian anyons with bound states and scattering states}

\author[1]{Sven Bachmann}
\author[2]{Bruno Nachtergaele}
\author[3]{Siddharth Vadnerkar}

\address[1]{Department of Mathematics, University of British Columbia, Vancouver, BC V6T 1Z2, Canada}
\address[2] {Department of Mathematics and Center for Quantum Mathematics and Physics, University of California, Davis, Davis, CA, 95616, USA}
\address[3]{Department of Physics, University of California, Davis, Davis, CA, 95616, USA}
\date{\today}                 

\maketitle
 
\begin{abstract} 
We introduce a family of quantum spin Hamiltonians on $\bbZ^2$ that can be regarded as perturbations of Kitaev's abelian quantum double models that preserve the gauge and duality symmetries of these models. We analyze in detail the sector with one electric charge and one magnetic flux and show that the spectrum in this sector consists of both bound states and scattering states of abelian anyons. Concretely, we have defined a family of lattice models in which abelian anyons arise naturally as finite-size quasi-particles with non-trivial dynamics that consist of a charge-flux pair. In particular, the anyons exhibit a non-trivial holonomy with a quantized phase, consistent with the gauge and duality symmetries of the Hamiltonian.
\end{abstract}


\section{Introduction}   \label{sec: introduction}

When Leinaas and Myrheim published their remarkable paper \cite{leinaas:1977}, the possibility of a new type of quasi-particles in two space dimensions with more general statistics than the well-established fermions and bosons, was a purely theoretical observation. This changed abruptly with the discovery of the fractional quantum Hall effect (FQHE) by Tsui, Stormer, and Gossard \cite{tsui:1982} and Laughlin's many-body wave function \cite{laughlin:1983} as a proposed explanation shortly after the experimental discovery. Wilczek called these hypothetical particles with fractional statistics {\em anyons} \cite{wilczek:1982} and they quickly became a central piece of many theoretical works on the FQHE \cite{haldane:1983b,tao:1983,halperin:1984,haldane:1985,jain:1990,haldane:1990,frohlich:1997}. Convincing experimental evidence for the existence of anyonic excitations in fractional quantum Hall materials has emerged in recent experiments \cite{bartolomei:2020,nakamura:2020}. 

All this motivated several efforts to construct explicit Hamiltonian models of anyons and study their mathematical and physical properties. Many of such efforts focus on interpreting anyons as charge and flux tube pairs \cite{lundholm:2014, lambert:2022}. A relation between the filling fraction and the ground state degeneracy was proved in \cite{LatticeAnyons,RationalIndex} for lattice models and in \cite{jansen:2009} for the Laughlin state. Haldane's pseudo-potential Hamiltonians have been derived in \cite{seiringer:2020}. A truncated version of a pseudo-potential for the $\nu=1/3$ filled Laughlin state introduced in \cite{nakamura:2012,bergholtz:2006} was shown to have a non-vanishing ground state gap in \cite{nachtergaele:2021a,warzel:2022}. Other aspects were investigated in a variety of models \cite{yakaboylu:2020, RationalIndex, rougerie:2023a}.

The possibility of topological quantum computation envisaged in~\cite{nayak2008non} has been a further incentive to deepen our understanding of anyons. Among the Hamiltonian models in two dimensions developed for that purpose, the most famous ones are Kitaev's quantum double models \cite{kitaev:2003} and his honeycomb model \cite{kitaev2006anyons}, as well as the Levin-Wen string-net models~\cite{levin2005string}. They do exhibit localized excitations with anyonic statistics, can be extended to describe anyon condensation~\cite{Bombin2008-uw,Levin-Wen_Condensation}, but these quasi-particles are not dynamical. In fact, the main purpose of these models is to realize, in their ground state space, a topological quantum field theory~\cite{witten1989quantum}. This has led to further developments in the classification of topological matter in the language of tensor categories~\cite{kong2014anyon,johnson2022classification}. Experimental realizations of these exotic topological phases have recently been proposed~\cite{verresen2021prediction,slagle2022quantum}.

These rigorous results recover some of the features of anyons in the theories for the FQHE mentioned above but it is fair to say that we are still lacking a coherent theory starting from first principles. In this paper we present Hamiltonian lattice models of anyons with short-range interactions of the type prevalent in condensed matter theories of real materials. It may also be realizable as an artificial system created in today's laboratories. We introduce our models as perturbations of Kitaev's quantum double models based on an arbitrary abelian group $G$ and rely on the stability of its anyon structure under short-range perturbations proved in \cite{cha:2020}. This generalizes earlier work on the case of $G=\bbZ_2$, which starts with the Toric Code model \cite{nachtergaele:2020}. These models have a unique gapped ground state on the infinite lattice, and the excitation spectrum can be described by a set of abelian anyons. They have a duality symmetry which is reflected in the structure of the anyons: they consist of two basic sets, which we call electric and magnetic charges that appear symmetrically in the Hamiltonian. Anyons that have both non-zero electric and magnetic charges obtained by fusion, appear with their own dispersion relation. We calculate the dispersion relation for single-anyon excitations and show that it exhibits a Dirac cone for these fused anyons with non-trivial braiding statistics.  We show that under a simple condition on the coupling constants and the energy and momentum, bound states appear that should be interpreted as finite-size quasi-particles with non-trivial dynamics. We also prove that bounds states do not appear when the Hamiltonian only includes separate hopping terms for electric and magnetic charges. In summary, we have constructed a family of lattice models in which abelian anyons arise naturally as charge-flux pairs with the expected properties. We explicitly find that the anyons exhibit a non-trivial holonomy with a quantized phase, consistent with the gauge and duality symmetries of the Hamiltonian. There is no need to introduce charge-flux pairs or tracer particles a priori \cite{lundholm:2016,yakaboylu:2020,lambert:2022}. The anyonic holonomy as observed in \cite{nakamura:2020} arises directly from the many-body Hamiltonian.

In Section \ref{sec: model} we introduce the family of Hamiltonians, discuss their symmetry properties, and describe the structure of the anyons and the associated superselection sectors for these models, including concrete instances of the corresponding GNS representations. The spectral analysis of the Hamiltonian in the invariant subspaces of the relevant GNS representations is carried out in Section \ref{sec: Spectrum}. We illustrate our findings with a set of figures based on numerical calculations.


\section{The model}   \label{sec: model}

\subsection{The (static) quantum double model}

Let us start with a quick reminder of the family of quantum double models (QDMs) with a finite Abelian group $(G,\cdot)$, which will allow us to set notations. We denote by $\Gamma$ the cell complex given by a set $\mathcal{V} = \bbZ^2$ of vertices, a set $\mathcal{E}$ of oriented edges and a set $\mathcal{F}$ of faces. We denote by $\partial$ the standard boundary map of this cell complex. For simplicity, all `vertical' edges are oriented upwards while all `horizontal' edges are pointing rightward. We attach to each edge $e\in\caE$ the $\vert G\vert$-dimensional Hilbert space $\mathcal{H}_e = \mathbb{C}^{\vert G\vert}$ and label an orthonormal basis by group elements $g \in G$: $\{\ket{g}_e:g\in G\}$; If a statement is not specific to one particular edge, we shall simply denote $\caH$. The observable algebra associated with each edge is $\cstar_e = \caL(\caH_e)$. For any finite subset $\Lambda\subset\caE$, the algebra is given by $\cstar_\Lambda = \caL(\otimes_{e\in\Lambda}\caH_e)$. As usual, $\Lambda_1\subset\Lambda_2$ yields a natural embedding $\cstar_{\Lambda_1}\hookrightarrow\cstar_{\Lambda_2}$ by the identification of $A\in\cstar_{\Lambda_1}$ with $A\otimes\idtyty_{\cstar_{\Lambda_2\setminus\Lambda_2}}$. The algebra of local observables $\cstar_{\mathrm{loc}}$ is the union of these finite volume algebras, and the C*-algebra of quasi-local observables $\cstar$ is the completion of $\cstar_{\mathrm{loc}}$ with respect to the operator norm.

For any $h\in G$, let $L^h$ be the linear operator defined by its action on the basis:
\begin{equation*}
L^h\ket{g} = \ket{hg}.
\end{equation*}
Similarly, for any character $\chi \in \hat G$, let $T^\chi$ be the linear operator defined by
\begin{equation*}
T^\chi\ket{g} = \bar\chi(g)\ket{g},
\end{equation*}
where $\bar \chi$ denotes the inverse of the character $\chi$. These operators form a unitary representation of $G$, respectively $\hat G$, on $\caH$, and $(L^h)\str = L^{\bar h}, (T^\chi)\str = T^{\overline \chi}$, where ${\bar h}=h^{-1}$, the inverse in $G$. Moreover, they satisfy the following commutation relation:
\begin{equation}\label{eqn:local_relations}
T^\chi L^h = \bar\chi(h)L^hT^\chi.
\end{equation}

It will be convenient to label operators not only by an edge, but by pairs $(v,e)$ with $v\in \partial e$ or $(f,e)$ with $e\in \partial f$. For a vertex-edge pair, $\caL^h(v,e) = L^h_e$ and $\caT^\chi(v,e) = T^\chi_e$ if the edge is outgoing from $v$, while $\caL^h(v,e) = L^{\bar h}_e$ and $\caT^\chi(v,e) = T^{\bar\chi}_e$ if it is incoming. Similarly for a face-edge pair with incoming (respectively outgoing) being replaced by the face being on the left (respectively on the right) of the edge. These conventions are illustrated in Figure~\ref{fig:convention}. 

\begin{figure*}[ht]
    \includegraphics[width = 0.8\textwidth]{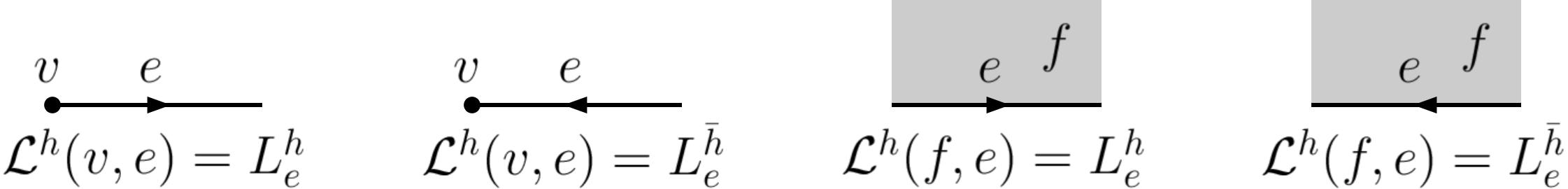}
    \caption{The $\caL$ operators associated with pairs $(v,e)$ and $(f,e)$. A parallel definition is made for the $\caT$ operators.}
    \label{fig:convention}
\end{figure*}

A string $\gamma$ is a sequence of pairs $\{(v_i,e_i):v_i\in\caV, e_i\in\caE\}$ such that $\partial e_i = \{v_i,v_{i+1}\}$. If it is finite, the last vertex is not listed but it is uniquely determined by $(v_N,e_N)$ and we denote it $\partial_1\gamma$; analogously, $\partial_0\gamma = v_1$, and we let $\partial\gamma = \{\partial_0\gamma,\partial_1\gamma\}$. We define similarly a dual string $\overline\gamma$, where faces replace vertices, and the edge $e_i$ is the unique edge in $\partial f_i\cap \partial f_{i+1}$, see Figure~\ref{fig:Strings}.

\begin{figure*}[ht]
    \includegraphics[width = 0.55\textwidth]{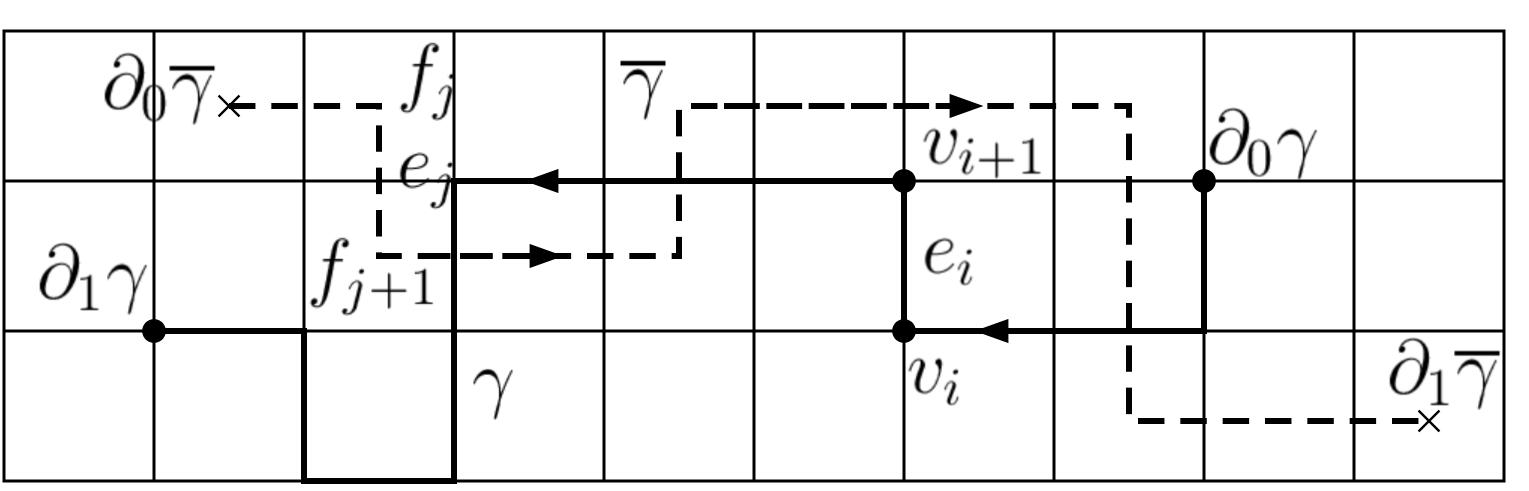}
    \caption{A geometric string (solid line) and dual string (dashed line). The ordering of the vertices/faces defines an orientation of strings, and correspondingly an initial and final vertex/face.}
    \label{fig:Strings}
\end{figure*}

 We now define the following unitary string operators
\begin{equation}\label{String ops}
       \caF^g_{\overline\gamma} = \underset{(f,e) \in \overline{\gamma}}{\prod} \mathcal{L}^g (f ,e)\qquad
    \caF^{\chi}_{\gamma}= \underset{(v,e) \in \gamma}{\prod} \mathcal{T}^{\chi} (v ,e)
\end{equation}
for $g\in G, \chi \in \hat G$. Note that the order in which the operators appear is irrelevant in the present context. The boundary of a face $f\in \mathcal{F}$ is naturally associated with a string $\gamma(f)$, while a vertex $v \in \mathcal{V}$ is naturally associated with a dual string $\overline \gamma(v)$. We shall always consider these strings to carry a counterclockwise orientation. We define for any face $f$ and vertex $v$ 
\begin{equation}
\face^h_f= \frac{1}{{\vert \hat G\vert}} \sum_{\chi \in \hat G} \chi(h) \caF^{\chi}_{\gamma(f)},
\qquad
\plus{\chi}_v= \frac{1}{{\vert G\vert}} \sum_{g \in G} \overline{\chi}(g)\caF^g_{\overline\gamma(v)}.
\end{equation}
The normalization is chosen so as to make these operators into (orthogonal) projections, and one can check that they are mutually commuting: In fact,
\begin{equation}\label{AB projections}
\caA_v^\chi\caA_v^\xi = \delta_{\chi,\xi}\caA_v^\chi = \caA_v^\xi\caA_v^\chi,\qquad
\caB_f^g\caB_f^h = \delta_{g,h}\caB_f^g = \caB_f^h \caB_f^g,
\end{equation}
by the orthogonality of the characters. Their action on basis vectors is illustrated in Figure~\ref{fig:AB_convention}.

\begin{figure}[htb]
    \includegraphics[width=\textwidth]{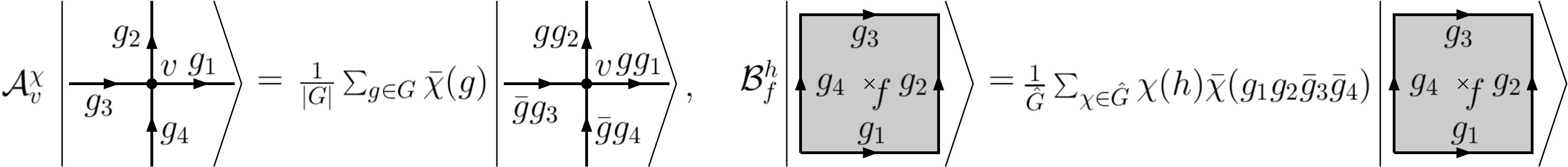}
    \caption{The action of the vertex and face operators.}
    \label{fig:AB_convention}
\end{figure}

For any finite subset $\Lambda\subset\caE$, the Hamiltonian of the (static) quantum double model is given by
\begin{equation}\label{QD}
    H^0_\Lambda = \underset{v\in\Lambda}{\sum} (\idtyty - \plus{\iota}_v) + \underset{f\in\Lambda}{\sum} ( \idtyty - \face^1_f),
\end{equation}
where $\iota\in\hat G $ is the constant character $\iota(g) = 1$, $1\in G$ is the neutral element, and $v\in\Lambda$ means that $\{e\in\caE:v\in\partial e\}\subset\Lambda$, and similarly for $f\in\Lambda$. This choice of boundary condition is arbitrary and made only for simplicity here; it will not play a role in the following discussion.

In fact, we shall be interested in the infinite-volume limit of the model. We recall some elementary properties of the model exhibited in the original~\cite{KitaevToricCode}, and refer to~\cite{Bombin2008-uw,SvenReview,AlickiFannesHorodecki,FiedlerPieter,Cha_GS} for a complete set of proofs. The sum~(\ref{QD}) is not convergent in norm, but it remains meaningful as the generator of the dynamics: for any local observable $A$, the limit
\begin{equation}\label{eq:derivation}
\delta^0(A) = \lim_{\Lambda\to\bbZ^2}\iu[H^0,A]
\end{equation}
exists for all $A\in\cstar_{\mathrm{loc}}$ and extends to a densely defined unbounded *-derivation on $\cstar$. In this infinite volume limit, the model has a unique translation-invariant ground state $\omega_0$ which is frustration-free and gapped -- we shall henceforth call it the vacuum state. $\omega_0$ is a ground state in the sense of being a positive normalized functional on $\cstar$ such that
\begin{equation}\label{GS condition}
-\iu\omega_0(A\str \delta^0(A))\geq0
\end{equation}
for all $A\in\cstar_{\mathrm{loc}}$. That it is gapped means that there is $g>0$ such that 
\begin{equation*}
-\iu\omega_0(A\str \delta^0(A))\geq g\omega_0(A\str A)
\end{equation*}
for all $A\in\cstar_{\mathrm{loc}}$ such that $\omega_0(A) = 0$. Specifically, $g=1$ here. Finally, the frustration-freeness is expressed  by
\begin{equation}\label{FF condition}
\omega_0(\plus{\iota}_v)=1,\qquad\omega_0(\face^1_f) = 1,
\end{equation}
for all $v\in\caV,f\in\caF$. 

Additional properties of the vacuum state are that $\omega_0(\plus{\chi}_v)=0=\omega_0(\face^g_f)$ for all non-trivial elements of $G$ (respectively $\hat G$) from which we deduce that
\begin{equation}\label{Invariance of vacuum}
\omega_0(\iu[\caA_v^\chi,A])=0,\qquad \omega_0(\iu[\caB_f^g,A])=0,
\end{equation}
for all $A\in\cstar$ by the Cauchy-Schwarz inequality.

Let us briefly come back to the string operators introduced above. They satisfy the following commutation relations:
\begin{equation}
\label{eqn:string_proj_algebra}
\begin{split}
& \plus{\chi}_{\partial_0 \gamma} \eigstr{\xi}{\gamma} = \eigstr{\xi}{\gamma} \plus{\chi\bar\xi}_{\partial_0 \gamma}, \qquad \plus{\chi}_{\partial_1 \gamma} \eigstr{\xi}{\gamma} = \eigstr{\xi}{\gamma} \plus{\chi\xi}_{\partial_1 \gamma}, \\  
& \face^h_{\partial_0 \overline{\gamma}} \dstr{g}{\gamma} = \dstr{g}{\gamma}\face^{h\overline{g}}_{\partial_0 \overline{\gamma}},
\qquad \face^h_{\partial_1 \overline{\gamma}} \dstr{g}{\gamma} =  \dstr{g}{\gamma} \face^{hg}_{\partial_1 \overline{\gamma}},
\end{split}
\end{equation}
where we used that $\chi(\bar g) = \overline{\chi(g)} = \overline\chi(g)$. Moreover, $\plus{\chi}_v$ commutes with $\eigstr{\xi}{\gamma}$ for $v \notin \partial \gamma$ and $\face^h_f$ commutes with $\dstr{g}{\gamma}$ for $ f \notin \partial \overline{\gamma}$. We further note that any (dual) string operators commute with all $\caA$'s and all $\caB$'s. We shall also need commutation relations between string operators of different types:
\begin{equation}\label{strings crossing}
\caF_\gamma^\chi \caF_{\overline\gamma}^g = (\chi(g))^{c(\gamma,\overline\gamma)}\caF_{\overline\gamma}^g \caF_\gamma^\chi 
\end{equation}
Here, $c(\gamma,\overline\gamma)$ is the number of signed crossings of the two strings: each crossing, from right to left, of $\overline\gamma$ and $\gamma$ yields a factor $\chi(g)$; since a change in the orientation of a string amounts to the exchange of the charge with its inverse, a crossing from left to right yields $\overline\chi(g) = \chi(\overline g)$.

Localized excitations above the ground state are simple to characterize. These quasi-particles arise from the violation of the conditions~(\ref{FF condition}). Indeed, we shall say that a state $\omega$ carries an electric charge at vertex $v$ if $\omega(\plus{\iota}_v) = 0$, or a magnetic flux through face $f$ if $\omega_0(\face^1_f) = 0$. The commutation relations~(\ref{eqn:string_proj_algebra}) yield that
\begin{equation*}
\omega_0((F_\gamma^\chi)\str A_{\partial_0\gamma}^\xi F_\gamma^\chi)=\delta_{\xi,\chi},\qquad
\omega_0((F_\gamma^\chi)\str A_{\partial_1\gamma}^\xi F_\gamma^\chi)=\delta_{\xi,\overline\chi},
\end{equation*}
so that it is natural to interpret the string operator $\caF^{\xi}_{\gamma}$ as creating a pair of quasi-particles with opposite electric charges $(\chi,\overline\chi)$ at either ends of $\gamma$, while $\caF^{g}_{\overline\gamma}$ creates a pair of fluxes $(g,\bar g)$ on the two extremal faces of $\overline\gamma$. This is consistent with the fact that if $\gamma_1,\gamma_2$ are such that $\partial_1\gamma_1 = \partial_0\gamma_2$, then $\caF^{\xi}_{\gamma_1}\caF^{\xi}_{\gamma_2} = \caF^{\xi}_{\gamma_1\cdot\gamma_2}$, where $\gamma_1\cdot\gamma_2$ denotes the concatenation of two strings.

Since this Hamiltonian is a sum of strictly local commuting terms, we have that
\begin{equation}\label{no dynamics}
\delta^0(\plus{\chi}_v) = 0 = \delta^0(\face^g_f),
\end{equation}
for all $\chi\in\hat G$ and $g\in G$, showing that excitations are not dynamical. We shall in the next section introduce hopping terms for these quasi-particles which will define a dynamical version of the QDM. 

\subsection{The dynamical quantum double model}

Hopping of electric charges between two vertices along an edge $e\in\caE$ will be induced by the following operator
\begin{equation}
T_e^\epsilon = \sum_{\chi\in\hat G,\chi\neq\iota} \left(
\mathcal{F}^{\overline{\chi}}_{(v,e)} \plus{\iota}_{v'} \plus{\chi}_{v} + \mathcal{F}^{\overline\chi}_{(v',e)} \plus{\iota}_{v} \plus{\chi}_{v'}
\right)
\end{equation}
where the edge $e$ is between $v$ and $v'$. Since $(\mathcal{F}^{\overline{\chi}}_{(v,e)})\str = \mathcal{F}^{\chi}_{(v,e)}=\mathcal{F}^{\overline{\chi}}_{(v',e)}$, the second term is the adjoint of the first one and hence $T_e^\epsilon$ is self-adjoint. Flux hopping between two faces across $e\in\caE$ can be defined similarly,
\begin{equation}
T_e^\mu = \sum_{h\in G,h\neq 1} 
\left(
\mathcal{F}^{\overline{h}}_{(f,e)} \face^{1}_{f'} \face^{h}_{f} + \mathcal{F}^{\bar h}_{(f',e)} \face^{1}_{f} \face^{h}_{f'}
\right)
\end{equation}
where the edge $e$ is shared by $f$ and $f'$.

Finally, and in accordance with the fact that a flux-charge pair can be considered a particle of its own type, we also introduce a hopping such a pair:
\begin{align}
T_{e}^{\epsilon\mu} = \sum_{\substack{h\in G,h\neq 1 \\ \chi\in\hat G,\chi\neq\iota}} \Big\{
&\left(
\mathcal{F}^{\overline{\chi}}_{(v,e)} \plus{\iota}_{v'} \plus{\chi}_{v} + \mathcal{ F}^{\overline \chi}_{(v',e)} \plus{\iota}_{v} \plus{\chi}_{v'}
\right)\left(\caB_f^h + \caB_{f'}^h\right) 
\nonumber \\
&+ \left(
\mathcal{F}^{\bar{h}}_{(f,e)} \face^{1}_{f'} \face^{h}_{f} + \mathcal{F}^{\bar h}_{(f',e)} \face^{1}_{f} \face^{h}_{f'}
\right)\left(\caA_v^\chi + \caA_{v'}^\chi\right)
\Big\}\label{em hopping}
\end{align}
where $e,v,v',f,f'$ are illustrated in Figure~\ref{fig:hopping of pair}.

\begin{figure}[htb]
    \includegraphics[width=0.2\textwidth]{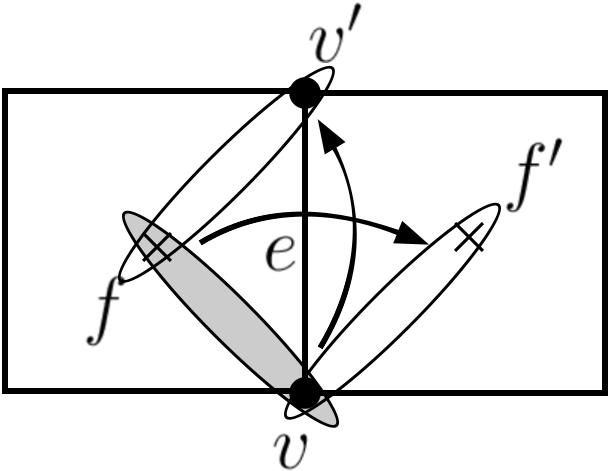}
    \caption{Two possible hoppings of a composite particle across or along a given edge $e$.}
    \label{fig:hopping of pair}
\end{figure}

With these definitions, we define the full Hamiltonian of the model by
\begin{equation*}
H = H^0 + \lambda_\epsilon H^\epsilon + \lambda_\mu H^\mu + \lambda_{\epsilon\mu} H^{\epsilon\mu},
\end{equation*}
where $H^\sharp = \sum_{e\in\caE} T^\sharp_e$, with $\sharp$ stands for $\epsilon$, $\mu$ or $\epsilon\mu$, and the sum is interpreted as in~(\ref{eq:derivation}). We denote
\begin{equation*}
\delta = \lim_{\Lambda\to\bbZ^2}\iu[H_\Lambda,\cdot] = \delta^0+\delta^\epsilon+\delta^\mu+\delta^{\epsilon\mu},
\end{equation*}
the generator of the dynamics in the infinite volume limit, while $\delta^\sharp$ are the derivations associated with $H^\sharp$. The corresponding dynamics is denoted $\tau_t$, respectively $\tau_t^\sharp$. 

The following proposition justifies the nomenclature that has been used thus far in describing the various terms of the Hamiltonian. The first part~(\ref{cont eq}) is a continuity equation for the charge $\zeta\in\hat G$ and it should be contrasted with~(\ref{no dynamics}). The operator $\caN_\Lambda^\zeta$ has the natural interpretation of the total charge of type $\zeta$ in the volume $\Lambda$. The second part of the proposition shows that the hopping terms are gauge invariant. 

\begin{prop}\label{prop:continuity equation}
Let $\zeta\in\hat G, \zeta\neq\iota,$ and let $\Lambda$ be a finite subset of $\caE$. Let $\caN_\Lambda^\zeta = \sum_{v\in\Lambda}\caA^\zeta_v$. \\
(i) Let $\caJ^\zeta_{(v,e)} = \iu\left(\caF^{\bar\zeta}_{(v,e)}\caA^\iota_{v'}\caA^\zeta_v - \caF^{\bar\zeta}_{(v',e)}\caA^\iota_{v}\caA^\zeta_{v'}\right)$. Then
\begin{equation}\label{cont eq}
\delta^\epsilon(\caN^\zeta_\Lambda) = \sum_{(v,e)\in\partial\Lambda}\caJ^\zeta_{(v,e)},
\end{equation}
where $(v,e)\in\partial\Lambda$ if $\partial e\cap\Lambda \neq \emptyset$ and $\partial e\cap\Lambda^c \neq \emptyset$. Moreover, $\delta^\mu(\caN^\zeta_\Lambda) = 0$. \\
(ii) For any edge $e\in\caE$, 
\begin{equation*}
\sum_{v\in\caV}\left[T_e^\epsilon, \caA^\zeta_v\right] = 0.
\end{equation*}
\end{prop}
\begin{proof}
The operators $T_e^\epsilon$ and $\caA^\zeta_v$ do not commute only if $v$ is one of the vertices at the boundary of $e$, in which case as short calculation yields
\begin{align*}
[T_e^\epsilon, \caA^\zeta_v] &= \sum_{\chi\in\hat G}
\left(\left[\caF^{\overline \chi}_{(v,e)},\caA^\zeta_v\right] \caA^\iota_{v'}\caA^\chi_v + 
\left[\caF^{\overline \chi}_{(v',e)},\caA^\zeta_v\right] \caA^\iota_v\caA^\chi_{v'}\right) \\
&= \sum_{\chi\in\hat G}\left(\caF^{\overline \chi}_{(v,e)}(\caA^\zeta_v - \caA^{\zeta\chi}_v)\caA^\iota_{v'}\caA^\chi_v
+ \caF^{\overline \chi}_{(v',e)}(\caA^\zeta_v - \caA^{\zeta\overline\chi}_v)\caA^\iota_v\caA^\chi_{v'}
\right).
\end{align*}
With~(\ref{AB projections}) and since $\zeta\neq\iota$, we have that $\caA^{\zeta\chi}_v\caA^\chi_v = 0$ as well as $\caA^\zeta_v\caA^\iota_v=0$, and this further imposes that $\chi = \zeta$ in both remaining terms, namely $\iu [T_e^\epsilon, \caA^\zeta_v] = \caJ^\zeta_{(v,e)}$. Claim (ii) immediately follows since $[T_e^\epsilon, \sum_v\caA^\zeta_v] = \caJ^\zeta_{(v_1,e)} + \caJ^\zeta_{(v_2,e)} = 0$ by the antisymmetry of $\caJ^\zeta_{(v,e)}$ under the exchange of the vertices. Claim (i) follows similarly: By the same argument, the contribution from all edges that are internal to $\Lambda$ vanish, the only remaining ones being those on the boundary. Finally, $\delta^\mu(\caN^\zeta_\Lambda) = 0$ follows from the fact that both $\caA$'s and dual strings are products of $\caL$ operators, which are all mutually commuting. 
\end{proof}

We conclude this section by pointing out that a natural extension of the Hamiltonian above would be to introduce an external `vector potential' by twisting the various hopping term, as for example
\begin{equation*}
\ep{-\iu\beta_e}\mathcal{F}^{\overline{\chi}}_{(v,e)} \plus{\iota}_{v'} \plus{\chi}_{v} + \ep{\iu\beta_e}\mathcal{F}^{\overline\chi}_{(v',e)} \plus{\iota}_{v} \plus{\chi}_{v'}.
\end{equation*}
An edge-dependent $\beta:\caE\to\bbR$ would allow one to add external `fluxes' through plaquettes. Of course, similar variations are possible with the other hopping terms. As we shall see shortly, such a flux attachment that is typical of anyons will appear intrinsically as soon as pairs of particles are present in the system.

\subsection{Superselection sectors and anyons}

Recall that the state $\omega_0$ is the unique translation invariant ground state. We briefly discussed above the states $\omega_0\circ\Gamma_\gamma^\chi$ where $\Gamma_\gamma^\chi = \caF_\gamma^{\overline{\chi}}(\cdot)\caF_\gamma^{{\chi}}$ creates a pair of excitations $(\chi,\overline\chi)$ localized at the ends of the string on top of the existing vacuum. 

If $\gamma = \{(v_i,e_i):i\in\bbN\}$ is an infinite string, the operator $\caF_\gamma^{\chi}$ is not well-defined, but the corresponding automorphism $\Gamma_\gamma^\chi$ exists. Indeed, let $\gamma^{(N)} = \{(v_i,e_i):i=1,\ldots,N\}$ be its truncation to the first $N$ edges. For any $A\in\cstar_{\mathrm{loc}}$, the limit
\begin{equation*}
\Gamma^\chi_\gamma(A) = \lim_{N\to\infty}\caF_{\gamma^{(N)}}^{\overline\chi} A \caF_{\gamma^{(N)}}^{\chi}
\end{equation*}
exists since the sequence of operators on the right hand side is eventually constant, and $\Gamma_\gamma^\chi$ extends to an automorphism of $\cstar$ since $\Vert \Gamma_\gamma^\chi(A)\Vert = \Vert A\Vert$. In particular, the linear functional $\omega_0\circ\Gamma^\chi_\gamma$ defines a state on $\cstar$, which we shall refer to as a charged state. It is so that 
\begin{equation*}
(\omega_0\circ\Gamma^\chi_\gamma)(\caA_v^\iota) = \begin{cases}
0 & \text{if }v=\partial_0\gamma \\ 1 & \text{otherwise}
\end{cases}\quad\text{and}\quad (\omega_0\circ\Gamma^\chi_\gamma)(\caB_f^1) = 1\:\text{ for all }f\in\caF.
\end{equation*}
Now, if $\lambda$ is a (finite) closed loop, namely $\partial_0\lambda = \partial_1\lambda$, then the action of $F_\lambda^\chi$ is trivial on the vacuum state, namely
\begin{equation}\label{closed loops are trivial}
\omega_0(\caF_{\lambda}^{\overline\chi} A \caF_{\lambda}^{\chi}) = \omega_0(A)\qquad\text{($\lambda$ closed)}
\end{equation}
for all $A\in\cstar$. This, and the composition property of string operators implies that 
\begin{equation*}
\omega_0\circ\Gamma^\chi_\gamma = \omega_0\circ\Gamma^\chi_{\tilde \gamma}
\end{equation*}
for all strings $\gamma$ such that $\partial\gamma = \partial\tilde\gamma$. If $\gamma$ is an infinite string, this shows that the state $\omega_0\circ\Gamma^\chi_\gamma$ depends only on the position of the initial vertex $v = \partial_0\gamma$, which justifies the notation $\omega_v^\chi$. Physically, the state $\omega_v^\chi$ is a state carrying one electric charge $\chi$ at the vertex $v$ and remains indistinguishable from the vacuum state everywhere else. It can be shown (\cite{FiedlerPieter},\cite{Cha_GS}) that $\omega_v^\chi$ is a ground state in the sense of~(\ref{GS condition}) for any $\chi\in\hat G$ and any $v\in\caV$.

Let $(\caH_0,\pi_0,\Omega_0)$ be the GNS representation of the vacuum state --- we shall henceforth refer to it as the vacuum representation. 
For a given state, all GNS representations are unitarily equivalent.
A charged state $\omega_v^\chi$ is {\em not} equivalent to $\omega_0$. Its GNS representation can be given by $(\caH_0,\pi_0\circ\Gamma_\gamma^\chi,\Omega_0)$, where $\partial_0\gamma = v$:
\begin{equation*}
\omega_v^\chi(A) = \langle \Omega_0|(\pi_0\circ\Gamma_\gamma^\chi)(A) \Omega_0\rangle = \omega_0\circ\Gamma_\gamma^\chi(A).
\end{equation*}
 Similarly, if $\chi\neq\xi$, then $\omega_v^\chi$ and $\omega_v^\xi$ are inequivalent.

While the state $\omega_v^\chi$ depends only on the endpoint $v$ of the string $\gamma$, the GNS representation does depend on the entire string $\gamma$. The uniqueness of the GNS representation implies that for any two strings $\gamma_1,\gamma_2$ starting at $v$, the representations $\pi_0\circ\Gamma_{\gamma_1}^\chi,\pi_0\circ\Gamma_{\gamma_2}^\chi$ are equivalent: this is very reminiscent of the freedom of choice of the flux line needed to introduce a flux piercing the plane in a 2d electron gas. We shall make a choice which will be convenient for the rest of the discussion: The string associated with the state $\omega_v^\chi$ is taken to be extending vertically down from the vertex~$v$, as illustrated in Figure~\ref{fig:gauge choice}.

The construction just discussed can be carried out with dual strings to yield charged ground states of the form $\omega_{f}^g = \omega_0\circ\Gamma_{\overline\gamma}^g$, and we pick those strings to be oriented vertically upwards from the face $f$.

\begin{figure}[htb]
    \includegraphics[width=0.25\textwidth]{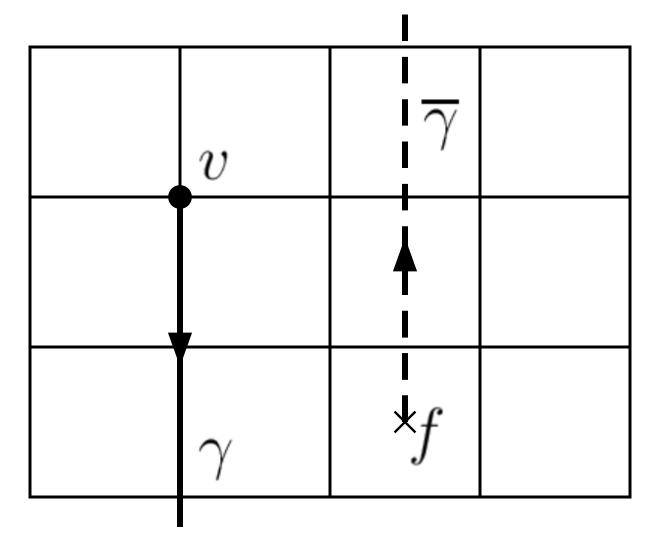}
    \caption{The gauge fixing for both types of states $\omega^\sharp_v$ and $\omega^\sharp_f$.}
    \label{fig:gauge choice}
\end{figure}

A superselection sector of $\delta^0$ is a class of unitarily equivalent $\delta^0$-ground state representations. The above discussion shows that group elements and characters label superselection sectors. We claim furthermore that states carrying the same charge but located at different vertices are equivalent. Indeed, the composition property of the string operators implies that two states $\omega_v^\chi$ and $\omega_{v'}^\chi$ are related by
\begin{equation}\label{algebraic transporters}
\omega_{v'}^\chi(A) = \omega_v^\chi(\caF^{\overline \chi}_{\gamma(v'\to v)} A \caF^{ \chi}_{\gamma(v'\to v)}),
\end{equation}
where $\gamma(v'\to v)$ is any string such that $\partial_0\gamma(v'\to v) = v'$ and $\partial_1\gamma(v'\to v) = v$. The operator $\caF^{ \chi}_{\gamma(v\to v')}$ is a unitary element of the algebra, so that the two states are indeed unitarily equivalent. That these are all possible sectors is the following result of~\cite{Cha_GS}:
\begin{thm}
The complete set of superselection sectors of $\delta^0$ is labelled by elements of $G\times\hat G$. 
\end{thm}

As was already pointed out in the original paper \cite{KitaevToricCode}, the quasi-particles associated with each superselection sector are Abelian anyons. In the present language, this can be made precise in the sense of the DHR analysis~\cite{DHR1}, see again~\cite{Cha_GS}. While the anyonic nature of particles has a general abstract treatment in terms of the superselection sectors, see~\cite{buchholzFredenhagen}, or~\cite{PieterDHR} for the quantum double models, we can understand braiding here in the sense of parallel transport, close to the original description of Leinaas and Myrheim, namely by starting with the two-particle representation $(\caH_0,\pi_0\circ\Gamma_{\overline\gamma}^g\circ\Gamma_\gamma^\chi,\Omega_0)$, and continuously braiding the electric charge $\chi$ around the flux $g$. For this, let $\lambda$ be a closed string of length $N$ such that $\partial_0\lambda = \partial_1\lambda = \partial_0\gamma$. Let $\Psi_\lambda(t)\in\caH_{0}$ be defined by $\Psi_\lambda(0) = \Omega_0$ and $\Psi_\lambda(t) = F^\chi_\lambda(t) \Omega_0$, where $F^\chi_\lambda(t)$ is given by
\begin{equation}
F^\chi_\lambda(t) = (\pi_0\circ\Gamma_{\overline\gamma}^g\circ\Gamma_\gamma^\chi) \left(\prod_{j=0}^{N} \Big(\idtyty + f(t - j)(\mathcal{F}^{\chi}_{(v_{N-j},e_{N-j})} - \idtyty)\Big)\right),\qquad(t>0).
\end{equation}
Here, $f$ is a smooth switch function such that $f(t) = 0$ if $t\leq0$ and $f(t) = 1$ if $t\geq1$. At integer time steps $t=k$ with $1\leq k\leq N$, we have that 
\begin{equation*}
\Psi_\lambda(k) = (\pi_0\circ\Gamma_{\overline\gamma}^g\circ\Gamma_\gamma^\chi) \left(\prod_{j=0}^{k-1} \mathcal{F}^{\chi}_{(v_{N-j},e_{N-j})} \right)\Omega_0,
\end{equation*}
and $\Psi_\lambda(t) = (\pi_0\circ\Gamma_{\overline\gamma}^g\circ\Gamma_\gamma^\chi) (\mathcal{F}^{\chi}_{\lambda}) \Omega_0$ for all $t\geq N$. It is now immediate to compute the holonomy associated with $\lambda$ and note that it is quantized.
\begin{thm}
Let $\eta_\lambda = \langle \Omega_0|\Psi_\lambda(N)\rangle$. Then
\begin{equation*}
\eta_\lambda = \chi(g)^{w_{\partial_0\overline\gamma}(\lambda)},
\end{equation*}
where $w_{\partial_0\overline\gamma}(\lambda)$ is the winding number of the string $\lambda$ around the face $\partial_0\overline\gamma$.
\end{thm}
\begin{proof}
The commutation relation~(\ref{strings crossing}) implies that
\begin{equation*}
(\Gamma_{\overline\gamma}^g\circ\Gamma_\gamma^\chi)(\mathcal{F}^{\chi}_{\lambda}) 
= \Gamma_{\overline\gamma}^g(\mathcal{F}^{\chi}_{\lambda}) = \chi(g)^{w_{\partial_0\overline\gamma}(\lambda)}\mathcal{F}^{\chi}_{\lambda}
\end{equation*}
which yields immediately $\eta_\lambda = \chi(g)^{w_{\partial_0\overline\gamma}(\lambda)}\langle \Omega_0,\pi_0(\mathcal{F}^{\chi}_{\lambda})\Omega_0\rangle = \chi(g)^{w_{\partial_0\overline\gamma}(\lambda)}$ since $\lambda$ is a closed string, see~(\ref{closed loops are trivial}). 
\end{proof}

\subsection{Symmetries} \label{sec:symmetries}

\subsubsection{Charge conservation}\label{sub:charge conservation}

For any $\chi\in\hat G$, let $\delta^\chi$ be the *-derivation on $\cstar$ formally given by $\iu\sum_{v}[\caA^\chi_v,\cdot]$. Similarly, for any $g\in G$, let $\delta^g$ be the *-derivation given by $\iu\sum_{f}[\caB_f^g,\cdot]$. Let $\alpha^\chi_s$, respectively $\beta^g_s$ be the corresponding groups of automorphisms. A slight extension of the proof of Proposition~\ref{prop:continuity equation} yields the following proposition:
\begin{prop}\label{prop:gauge invariance}
For all $\chi\in\hat G$ and all $g\in G$,
\begin{equation*}
\alpha^\chi_s \circ\tau_t = \tau_{t}\circ\alpha^\chi_s,\qquad 
\beta^g_s\circ\tau_t = \tau_{t}\circ\beta^g_s,
\end{equation*}
for all $s,t\in\bbR$.
\end{prop}
%

\subsubsection{Charge conjugations}

There are two types of conjugations in these models, corresponding to a flux-anti flux exchange or a charge-anti charge exchange. We first define an antilinear transformation $\theta^\mu$ on $\caH$ by
\begin{equation*}
\theta^\mu\vert g\rangle = \vert\bar g\rangle,
\end{equation*}
and extension by antilinearity, for which $(\theta^\mu)^2 = \idtyty$. It is immediate to check that
\begin{equation*}
\theta^\mu L^g = L^{\bar g} \theta^\mu,\qquad \theta^\mu T^\chi = T^\chi \theta^\mu,
\end{equation*}
for all $g\in G,\chi\in\hat G$. The automorphism
\begin{equation*}
\Theta^\mu(\cdot) = \lim_{\Lambda\to\bbZ^2}\Theta^\mu_\Lambda(\cdot)\Theta^\mu_\Lambda,
\end{equation*}
where $\Theta^\mu_\Lambda = \otimes_{e\in\caE_\Lambda}\theta_e$, acts as
\begin{equation*}
\Theta^\mu(\caA_v^\chi) = \caA_v^{\chi},\qquad \Theta^\mu(\caB^h) = \caB^{\bar h}.
\end{equation*}

A second antiunitary on $\caH$ is given by complex conjugation, namely
\begin{equation*}
\theta^\epsilon\vert g\rangle = \vert g\rangle,
\end{equation*}
and extension by antilinearity. It again squares to the identity but
\begin{equation*}
\theta^\epsilon L^g = L^{g} \theta^\epsilon,\qquad \theta^\epsilon T^\chi = T^{\overline \chi} \theta^\epsilon,
\end{equation*}
for all $g\in G,\chi\in\hat G$. It follows that
\begin{equation*}
\Theta^\epsilon( \caA_v^\chi ) = \caA_v^{\overline \chi},\qquad \Theta^\epsilon( \caB^h ) = \caB^h,
\end{equation*}
where $\Theta^\epsilon$ is the automorphism of $\cstar$ associated with the local $\theta^\epsilon$.

\begin{prop}\label{prop:discrete symmetries}
For all $\chi\in\hat G$ and all $g\in G$,
\begin{equation*}
\Theta^\sharp \circ\tau_t = \tau_{-t}\circ \Theta^\sharp,
\end{equation*}
where $\sharp\in\{\epsilon,\mu\}$.
\end{prop}
%

\subsubsection{Duality}

The Abelian quantum double models exhibit a duality symmetry that corresponds to the exchange of fluxes with charges. Using the Fourier transform associated with the finite Abelian group $G$, and the fact that $G$ and $\hat G$ are isomorphic as 
Abelian groups, we can define a duality transformation $U$ on the edge Hilbert space $\caH$. For clarity we recall the standard definitions. 
Since $\caH\cong \ell^2(G)\cong \ell^2(\hat G)$, and using the orthogonality relation for characters, we can introduce an orthogonal 
basis $\{\ket{\chi}\mid \chi\in\hat G\}$, labelled by the characters, by requiring $\langle\chi|\xi\rangle=\overline{\chi}(g)$. The conventional normalization 
then gives:
$$
\langle\chi|\xi\rangle=|G| \delta_{\chi,\xi}, \quad \chi,\xi\in\hat G.
$$
The two bases $\{\ket{g}: g\in G\}$ and $\{\ket{\chi}: \chi\in \hat G\}$ give rise to a unitary transformation $U$ on $\caH=\ell^2(G)$ such that 
$$
\hat v = Uv, \mbox{ with } \hat v(\chi) = \vert G\vert^{-1/2}\sum_{g\in G} \overline{\chi}(g) v(g).
$$
and with the inverse given by
$$
v(g) = \vert G\vert^{-1/2} \sum_{\chi\in\hat G} \chi(g) \hat v(\chi).
$$
As operators on $\ell^2(G)$, $L^h$ and $T^\chi$ take the form
$$
(L^h v)(g) = v(\bar h g), \quad (T^\chi v)(g) = \overline{\chi}(g)v(g).
$$
The Fourier transform interchanges their actions on $\ell^2(\hat G)$ (up to a complex conjugate)
$$
U L^h U^* \hat v (\xi) = {\xi}(h) \hat v (\xi), \quad U T^\chi U^* \hat v (\xi) = \hat v (\chi\xi).
$$
It follows that the transformed Hamiltonian $(\otimes U) H (\otimes U^*)$ is identical to $H$ upon the replacement of $G$ by $\hat G$ and the lattice by its dual lattice, and after interchanging $\lambda_\epsilon$ 
and $\lambda_\mu$.
 
Note that in the cyclic case $G=\bbZ_N$, with the characters labelled by $m=0,\ldots,N-1$, $\chi_m: j\mapsto \omega^{mj}$, with $\omega=e^{\frac{2\pi i}{N}}$, $U$ is the usual finite Fourier transform:
\begin{equation*}
U\vert j\rangle = \frac{1}{\sqrt N}\sum_{k=0}^{N-1}\omega^{jk}\vert k\rangle.
\end{equation*}
%

\subsection{Dynamics in Hilbert space}\label{sec:Dynamics}

The vacuum state remains invariant under the full dynamics $\tau_t$ since $\omega_0\circ\delta = \omega_0$. It follows that
$\tau_t$ is unitarily implementable in the vacuum representation, namely there is a strongly continuous unitary group $t\mapsto U_{0,t}$ on $\caH_0$ such that
\begin{equation}\label{U Implementability in vacuum}
U_{0,t}\pi_0(A)(U_{0,t})\str = \pi_0(\tau_t(A)),\qquad U_{0,t}\Omega_0 = \Omega_0.
\end{equation}
Let $K_0$ be its self-adjoint generator, normalized so that $K_0\Omega_0=0$:
\begin{equation}\label{Implementability in vacuum}
\iu[K_0,\pi_0(A)] = \pi_0(\delta(A)).
\end{equation}

While the charged states were also invariant under $\delta^0$, this is not the case anymore in the presence of the hopping terms. However, the dynamics is still unitarily implementable in the charged representations. We start with the charged state having one electric charge $\chi$. 

\begin{thm}
The dynamics $\tau_t$ is unitarily implementable in the GNS representation of $\omega_v^\chi$: There is a strongly continuous unitary group $t\mapsto U_{\gamma,t}^\chi$ on $\caH_0$ such that
\begin{equation*}
U_{\gamma,t}^\chi(\pi_0\circ\Gamma_\gamma^\chi)(A)(U_{\gamma,t}^\chi)\str = (\pi_0\circ\Gamma_\gamma^\chi)(\tau_t(A)),
\end{equation*}
for all $A\in\cstar$.
\end{thm}
\begin{proof}
Since the automorphism $\Gamma^\chi_\gamma$ is strictly local, in the sense that $\Gamma^\chi_\gamma(\caA_\Lambda)\subset\caA_\Lambda$ for any $\Lambda\subset\caE$, its action on any Hamiltonian is naturally  defined by letting $\Gamma^\chi_\gamma$ act on every interaction term. Denoting the resulting Hamiltonian $\Gamma^\chi_\gamma(H)$ and the corresponding derivation $\delta^{\Gamma^\chi_\gamma(H)}$, we have that
\begin{equation*}
\Gamma^\chi_\gamma\circ\delta = \delta^{\Gamma^\chi_\gamma(H)}\circ\Gamma^\chi_\gamma.
\end{equation*}
A short calculation yields that
\begin{equation}\label{change of static}
\Gamma_\gamma^\chi(H^0) = H^0 + (\caA^\iota_{\partial_0\gamma} - \caA^{\overline\chi}_{\partial_0\gamma})
\end{equation}
and
\begin{equation}\label{change of epsilon hopping}
\Gamma_\gamma^\chi(H^\epsilon) = H^\epsilon + \sum_{\substack{e\in\caE: \\ \partial e = \{\partial_0\gamma,v\}}}
\sum_{\xi\in\hat G,\xi\neq\iota}
\Big(\caF^{\overline\xi}_{(\partial_0\gamma,e)}\caA^{\iota}_{v}(\caA^{\overline\chi\xi}_{\partial_0\gamma} - \caA^{\xi}_{\partial_0\gamma})
+ \caF^{\overline \xi}_{(v,e)}(\caA^{\overline\chi}_{\partial_0\gamma} - \caA^{\iota}_{\partial_0\gamma})\caA^{\xi}_{v}\Big).
\end{equation}
The action of $\Gamma_\gamma^\chi$ on $H^\mu$ is trivial, up to the phase picked by all hopping terms across $\gamma$, see~(\ref{strings crossing}), resulting in
\begin{equation}\label{change of mu hopping}
\Gamma_\gamma^\chi(H^\mu) = H^\mu + \sum_{e\in\gamma}\sum_{h\in G,h\neq 1} 
\left(
(\chi(h)^{c(\gamma,(f,e))}-1) \mathcal{F}^{\overline{h}}_{(f,e)}\face^{1}_{f'} \face^{h}_{f} + \left(\chi(h)^{c(\gamma,(f',e))}-1\right)\mathcal{F}^{\bar h}_{(f',e)} \face^{1}_{f} \face^{h}_{f'}
\right).
\end{equation}
Finally, we consider the action of $\Gamma_\gamma^\chi$ on $H^{\epsilon\mu}$. The charge hopping corresponding to the first line of~(\ref{em hopping}) is affected as in~(\ref{change of epsilon hopping}) since $\caB$ operators are left invariant. The flux hopping corresponding to the second line of~(\ref{em hopping}) if affected in two ways: Firstly as in~(\ref{change of mu hopping}), the hopping terms across the string $\gamma$ pick up a phase, and secondly the projectors $\caA_{\partial_0\gamma}^\xi$ are mapped to $\caA_{\partial_0\gamma}^{\xi\overline\chi}$. 

By~(\ref{Invariance of vacuum}), all terms that contain an $\caA^\chi$ for $\chi\neq\iota$ or a $\caB^g$ for $g\neq 1$ leave the vacuum state invariant. The same holds for $H^0$, and we gather the action of all of them under the notation $\delta^\chi_{\gamma,0}$. The only terms that have a non-trivial action on $\omega_0$ are supported in a finite neighbourhood of $\partial_0\gamma$ and given explicitly by
\begin{equation*}
P^\chi_{\partial_0\gamma} = \caA^\iota_{\partial_0\gamma}
+ \sum_{\substack{e\in\caE: \\ \partial e = \{\partial_0\gamma,v\}}}
\caF^{\overline\chi}_{(\partial_0\gamma,e)}\caA^{\iota}_{v}\caA^{\iota}_{\partial_0\gamma}\in\cstar_{\mathrm{loc}}.
\end{equation*}
Summarizing, we have that
\begin{equation*}
\delta^{\Gamma^\chi_\gamma(H)} = \delta^\chi_{\gamma,0}+\iu[P^\chi_{\partial_0\gamma},\cdot],
\end{equation*}
and $\omega_0$ is invariant under the action of~$\delta^\chi_{\gamma,0}$. We conclude as in~(\ref{Implementability in vacuum}) that $\delta^\chi_{\gamma,0}$ is implemented by a self-adjoint operator $K_{\gamma,0}^\chi$ and so
\begin{equation}\label{chi GNS Hamiltonian}
\pi_0\circ\Gamma_\gamma^\chi\circ\delta(A) = \iu[K_{\gamma,0}^\chi + \pi_0(P^\chi_{\partial_0\gamma}),\pi_0\circ\Gamma_\gamma^\chi(A)]\qquad K_{\gamma,0}^\chi\Omega_0 = 0.
\end{equation}
Hence, $U_{\gamma,t}^\chi = \ep{\iu t K_\gamma^\chi}$ where $K_\gamma^\chi = K_{\gamma,0}^\chi + \pi_0(P^\chi_{\partial_0\gamma})$.
\end{proof}

The Hamiltonian $K^{\chi}_{\gamma}$ --- and its cousins to be described shortly --- will be the main object of the spectral analysis carried out in the next section. Before doing this, we recall that, in the representation $\pi_0\circ\Gamma_\gamma^\chi$, the vector $\Omega_0$ corresponds to the state having one electric charge $\chi$ at $\partial_0\gamma$. That this charge is dynamical is reflected in the fact that the propagator has a non-trivial action on $\Omega_0$ since 
\begin{equation}\label{Hopping in chi GNS}
K^{\chi}_{\gamma} \Omega_0 = \pi_0(P^\chi_{\partial_0\gamma}) \Omega_0
=\Omega_0 + \lambda_\epsilon \sum_{e:\partial_0\gamma\in\partial e}\pi_0\left(\caF^{\overline\chi}_{(\partial_0\gamma,e)}\right) \Omega_0,
\end{equation}
where we used that $\pi_0(\caA^{\iota}_{v})\Omega_0 = \Omega_0$ for all $v\in\caV$.

The results above extend with the obvious modifications to the case of a magnetic charge localized on a face. The case of both an electric charge $\chi\in\hat G$ at vertex $v$ and a magnetic flux $g\in G$ at face $f$ is similar but requires a little more attention. The GNS representation of the state $\omega^{\chi,g}_{v,f}$ is given by $(\pi_0\circ\Gamma^g_{\overline \gamma}\circ\Gamma_\gamma^\chi,\caH_0,\Omega_0)$, where $\partial_0\gamma = v$ and $\partial_0\overline\gamma = f$, as in Figure~\ref{fig:gauge choice}. We denote $K^{\chi,g}_{\gamma,\overline\gamma}$ the generator of the dynamics in the GNS representation. It is of the same form as~(\ref{chi GNS Hamiltonian}), namely
\begin{equation}\label{GNS Hamiltonian}
K^{\chi,g}_{\gamma,\overline\gamma} = K^{\chi,g}_{\gamma,\overline\gamma,0} + \pi_0\left(P^{\chi,g}_{v,f}\right),
\qquad K^{\chi,g}_{\gamma,\overline\gamma,0}\Omega_0 = 0,
\end{equation}
where $P^{\chi,g}_{v,f}$ is again strictly local. The action of $\Gamma_{\overline\gamma}^g$ on~(\ref{change of static}) is simple,
\begin{equation*}
(\Gamma^g_{\overline \gamma}\circ\Gamma^\chi_{\gamma})(H^0)-H^0 =
(\caA^\iota_{\partial_0\gamma} - \caA^{\overline\chi}_{\partial_0\gamma})
+(\caB_{\partial_0\overline\gamma}^1 - \caB_{\partial_0\overline\gamma}^g),
\end{equation*}
providing $(\caA^\iota_{\partial_0\gamma} + \caB_{\partial_0\overline\gamma}^1)$ to $P^{\chi,g}_{v,f}$. Secondly, $\Gamma_{\overline\gamma}^g$ adds phases to the hopping terms of~(\ref{change of epsilon hopping}) that cross $\overline\gamma$ with only 
\begin{equation}\label{eHops}
\sum_{\substack{e\in\caE: \\ e = \langle\partial_0\gamma,v\rangle}}
\chi(g)^{-c((\partial_0\gamma,e),\overline \gamma)}\caF^{\overline\chi}_{(\partial_0\gamma,e)}\caA^{\iota}_{v}\caA^{\iota}_{\partial_0\gamma}
\end{equation}
not leaving $\omega_0$ invariant. Furthermore, the only non-trivial action of $\Gamma_{\overline\gamma}^g$ on $\Gamma_{\gamma}^\chi(H^\mu)$ given in~(\ref{change of mu hopping}) is at the end of the dual string $\partial_0\overline\gamma$, which adds
\begin{equation}\label{muHops}
\sum_{\substack{e\in\caE: \\ e \in \partial\partial_0\gamma\cap\partial f}}\chi(g)^{c(\gamma,(\partial_0\overline\gamma,e))}\caF^{\overline g}_{(\partial_0\overline\gamma,e)}\face^{1}_{f} \face^{1}_{\partial_0\overline\gamma}.
\end{equation}
to $P^{\chi,g}_{v,f}$.
Finally, $(\Gamma^g_{\overline \gamma}\circ\Gamma^\chi_{\gamma})(H^{\epsilon\mu})$ acts that act non-trivially on $\omega_0$ if and only if $\partial_0\gamma, \partial_0\overline\gamma$ belong to the same site, namely $\partial_0\gamma \in\partial_0\overline\gamma$. In that case, there are four different possible hopping terms of either form above, as illustrated in Figure~\ref{fig:hoppinG phases}.

\begin{figure}[htb]
    \includegraphics[width=0.2\textwidth]{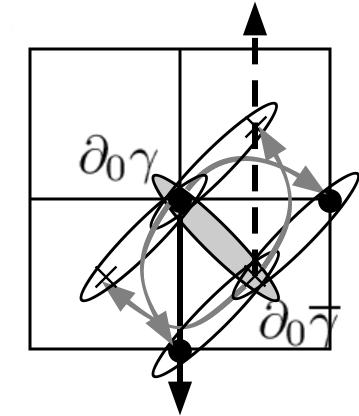}
    \caption{The four possible hopping terms arising from $H^{\epsilon\mu}$. A non-trivial phase is only associated with the two horizontal displacements.}
    \label{fig:hoppinG phases}
\end{figure}
%

As expected, the see that in a two-particle state, the $\chi$ particle exhibits nearest neighbour hopping on the set of vertices, the $g$ particle exhibits nearest neighbour hopping on the set of faces, and the hopping of the composite particle is to the neighbouring sites. The hopping of each type of single particle is influenced by the presence of the other type in its immediate vicinity.


\section{Spectral analysis in the one-particle sectors}   \label{sec: Spectrum}

\subsection{The Hamiltonian in one- and two-particle subspaces}

We introduced in Section~\ref{sub:charge conservation} the derivations corresponding to the number operators. In the GNS representations discussed above, the derivation is implemented by the commutator with a self-adjoint operator, analogous to the number operator on Fock space. Let us first recall the definition of~$\caN^\zeta_\Lambda$ in Proposition~\ref{prop:continuity equation}. We define
\begin{equation*}
\caN^{\hat G}_\Lambda = \sum_{\zeta\in\hat G} \caN^\zeta_\Lambda \in\cstar.
\end{equation*}
If $\partial_0\gamma\in\Lambda$, then
\begin{equation*}
\Gamma_\gamma^\chi\left(\caN^{\hat G}_\Lambda \right) = \caN^{\hat G}_\Lambda + \left(\caA^\iota_{\partial_0\gamma} - \caA^\chi_{\partial_0\gamma}\right).
\end{equation*}
In the vacuum representation, the limit
\begin{equation*}
N_0^{\hat G} = \wlim_{\Lambda\to\caV}\pi_0(\caN^{\hat G}_\Lambda)
\end{equation*}
defines an unbounded self-adjoint operator with a dense domain for which $N_0^{\hat G}\Omega_0 = 0$. It follows that the derivation $\lim_{\Lambda\to\caV}\iu[\caN_\Lambda^{\hat G},\cdot]$ is implemented in the GNS space of $\omega_{\partial_0\gamma}^\chi$ by
\begin{equation*}
N^{\hat G} = N_0^{\hat G} + \pi_0(A^\iota_{\partial_0\gamma} - A^\chi_{\partial_0\gamma}).
\end{equation*}
In particular, $N^{\hat G}\Omega_0 = \pi_0(A^\iota_{\partial_0\gamma})\Omega_0 = \Omega_0$. Had we not summed over $\hat G$, the same sequence of ideas would have given a generator $N^\zeta$ for any $\zeta\in\hat G$ with the property that $N^\zeta\Omega_0 = \delta_{\zeta,\chi}\Omega_0$.

Replacing $\caA^\chi$'s by $\caB^g$'s, we obtain a second number operator denoted $N^{ G}$ and such that $N^{ G}\Omega_0=0$. Finally,
\begin{equation*}
N = N^{\hat G} + N^{ G}
\end{equation*}
is the total number operator. 

Two charged states differing only by the position of the particle are unitarily equivalent and the intertwiner of representations is explicit. (\ref{algebraic transporters}) shows that if $T_{v\to v'} = \pi_0(\caF^{ \chi}_{\gamma(v'\to v)})$, then $T_{v\to v'}\str (\pi_0\circ\Gamma_{v}^\chi)(A) T_{v\to v'}$ provides a GNS representation of $\omega_{v'}^\chi$ on $\caH_0$. The corresponding string does however not satisfy the gauge fixing condition. As shown in~\cite{Cha_GS}, this can be remedied by introducing an additional unitary intertwiner $V^\chi$ which is explicitly given as a limit, in the weak operator topology, of operators corresponding to closed strings. In particular, $V^\chi\Omega_0 = \Omega_0$. Precisely,
\begin{equation*}
\Gamma^\chi_{\gamma'}(A) = \lim_{n\to\infty}(\caF_{v,v',n}^\chi)\str (\caF^{ \chi}_{\gamma(v'\to v)})\str \Gamma^\chi_\gamma(A)
\caF^{ \chi}_{\gamma(v'\to v)}\caF_{v,v',n}^\chi
\end{equation*}
for all $A\in\cstar_{\mathrm{loc}}$, and $\omega_0((\caF_{v,v',n}^\chi)\str A \caF_{v,v',n}^\chi) = \omega_0$ as well as $V ^\chi= \wlim_{n\to\infty}\pi_0(\caF_{v,v',n}^\chi)$, where $\caF_{v,v',n}^\chi$ are illustrated in Figure~\ref{fig:transporters}. We conclude that
\begin{equation*}
(\caH_0, \pi_0\circ\Gamma_\gamma^\chi, T_{v\to v'}\Omega_0)
\end{equation*}
is a GNS representation of $\omega_{v'}^\chi$, with the gauge fixed as prescribed. 

\begin{figure}[htb]
  \centering
  \includegraphics[width=0.275\linewidth]{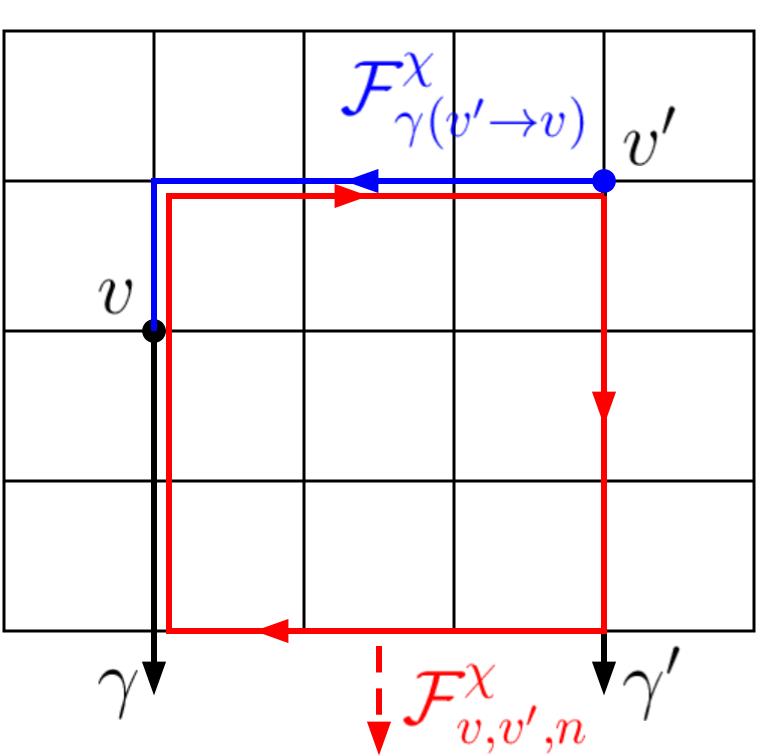}
  \caption{Transporting a charge from $v$ to $v'$ and respecting the gauge fixing. The sequence of unitaries $(\pi_0(\caF_{v,v',n}^\chi))_{n\in\bbN}$ has a good weak limit as the bottom stretch is sent to $\infty$.}
  \label{fig:transporters}
\end{figure}

For notational clarity and simplicity, we denote the vectors $T_{v\to v'}\Omega_0$ of $\caH_0$ in the representation $\pi_0\circ\Gamma_\gamma^\chi$ by $\{\vert \chi,v'\rangle:v'\in\caV\}$ and let $\caH_\chi\subset \caH_0$ be the space they span; in particular, $\Omega_0 = \vert \chi,v\rangle$. They satisfy 
\begin{equation*}
N\vert \chi,v\rangle = \vert \chi,v\rangle
\end{equation*}
for all $v\in \caV$ and we refer to them as one-particle states. The vector $\vert \chi,v\rangle$ corresponds to a state having exactly one charge $\chi$ localized at vertex $v$. The same holds for the dual picture, yielding one-particle vectors $\{\vert g,f'\rangle:f'\in\caF\}$ of $\caH_0$ in the representation $\pi_0\circ\Gamma^g_{\overline \gamma}$, for any $g\in G$.

Finally, in the GNS space of the states $\omega^{\chi,g}_{v,f}$, the corresponding vectors $\{\vert (\chi,v), (g,f)\rangle:v\in \caV,f\in\caF\}$ are such that $N\vert (\chi,v), (g,f)\rangle = 2\vert (\chi,v), (g,f)\rangle$ and we refer to their span $\caH_{\chi,g}\subset \caH_0$ as a two-particle subspace. In fact, $\caH_{\chi,g}$ is completely characterized by 
\begin{equation*}
\psi\in\caH_{\chi,g}\quad \Longleftrightarrow\quad N\psi = 2\psi\text{ and }N^\chi\psi = \psi,\,N^g\psi = \psi.
\end{equation*}
We note here, and shall use it later, that
\begin{equation*}
\caH_{\chi,g} \simeq\ell^2(\bbZ^2\times\bbZ^2).
\end{equation*}
We also point out that the transportability of charges remains simple in the two-particle state. Indeed, while $\caF$ operators corresponding to intersecting strings and dual strings do not commute, the corresponding automorphisms do so: $\Gamma_{\overline \gamma}^g((\caF_{v,v',n}^\chi)\str (\caF^{ \chi}_{\gamma(v'\to v)})\str \Gamma_\gamma^\chi(A) \caF^{ \chi}_{\gamma(v'\to v)}\caF_{v,v',n}^\chi) = (\caF_{v,v',n}^\chi)\str (\caF^{ \chi}_{\gamma(v'\to v)})\str (\Gamma^g_{\overline \gamma}\circ \Gamma_\gamma^\chi)(A) \caF^{ \chi}_{\gamma(v'\to v)}\caF_{v,v',n}^\chi$. Hence,
\begin{equation*}
(\caH_0, \pi_0\circ\Gamma^g_{\overline \gamma}\circ \Gamma_\gamma^\chi, T^\chi_{v\to v'}\Omega_0)
\end{equation*}
is a GNS representation of $\omega^{\chi,g}_{v',f}$, with the gauge fixed as prescribed. 

Since charge is conserved, we obtain the following:
\begin{prop}
Let $K^{\chi,g}_{\gamma,\overline\gamma}$ be the GNS Hamiltonian introduced in~(\ref{GNS Hamiltonian}). Then
\begin{equation*}
K^{\chi,g}_{\gamma,\overline\gamma}\caH_{\chi,g}\subset\caH_{\chi,g}.
\end{equation*}
\end{prop}

We now concentrate our attention on the Hamiltonian restricted to this two-particle (respectively one if either $\chi = \iota$ or $g = 1$) invariant subspace. The fact that the GNS representation and $K^{\chi,g}_{\gamma,\overline\gamma}$ are explicit allows us to compute all relevant matrix elements, allowing for a spectral study of the model.

\subsection{A matrix representation of the Hamiltonian}\label{sec: Hamiltonian warm up} 

To start with an easy case, we first consider~(\ref{Hopping in chi GNS}), denote again $v = \partial\gamma$, and compute
\begin{equation*}
\langle \chi,{v_1} | K_\gamma^\chi |\chi,{v}\rangle = \langle \chi,{v_1} | K_\gamma^\chi | \Omega_0 \rangle
=\langle \chi,{v_1} |\Omega_0\rangle + \lambda_\epsilon \sum_{e:v\in\partial e}\langle \chi,{v_1} |\pi_0(\caF^{\overline\chi}_{(v,e)}) \Omega_0\rangle
\end{equation*}
If $v\neq v_1$, then
\begin{equation*}
\langle \chi,{v_1} | \chi,{v}\rangle = \langle \pi_0(\caF^\chi_{\gamma(v_1\to v)})\Omega_0 |\Omega_0\rangle = \omega_0(\caF^\chi_{\gamma(v_1\to v)})
\end{equation*}
vanishes since $\pi_0(\caF^\chi_{\gamma(v_1\to v)})$ adds one pair of particles to the existing $\Omega_0$. More formally, since $\chi\neq\iota$, there is $g\in G$ such that $\chi(g)\neq 1$; the invariance of $\omega_0$ under the action of $\caF^g_{\overline{\lambda}}$ for any closed dual path $\overline{\lambda}$ and~(\ref{strings crossing}) together yield that for any $A\in\cstar$,
\begin{equation*}
\omega_0(\caF^\chi_{\gamma(v_1\to v)}) = \omega_0((\caF^g_{\overline{\lambda}})\str\caF^\chi_{\gamma(v_1\to v)}\caF^g_{\overline{\lambda}}) = 
\chi(g)\omega_0(\caF^\chi_{\gamma(v_1\to v)})
\end{equation*}
whenever $\overline{\lambda}$ winds once around $v$ but not around $v_1$, and so $\omega_0(\caF^\chi_{\gamma(v_1\to v)}) = 0$ indeed.

For the same reason, $\langle \chi,{v_1} |\pi_0(\caF^{\overline\chi}_{(v,e)}) \Omega_0\rangle$ is non-zero only if $\caF^\chi_{\gamma(v_1\to v)} = \caF^{\overline\chi}_{(v,e)}$ (in which case it is equal to $1$) namely $\gamma(v_1\to v) = (v_1,e)$. As could have been expected from the basic picture of a single particle freely hopping on the lattice, we conclude that the one-particle Hamiltonian reduces to a multiple of the discrete Laplacian:
\begin{equation*}
\langle \chi,{v_1} | K_\gamma^\chi |\chi,{v_2}\rangle=
\begin{cases}
\lambda_\epsilon & \text{if }d(v_1,v_2)=1, \\ 
1 & \text{if }v_1 = v_2,\\
0& \text{otherwise.}
\end{cases}
\end{equation*}
The same holds for any other type of particle, with the hopping amplitude adjusted accordingly. In the following, we shall renormalize the Hamiltonian by replacing $K_\gamma^\chi\to K_\gamma^\chi-\idtyty$.

In the two-particle situation, it will be convenient to insist on the representation of the lattice using sites. We first rescale the lattice so that vertices are on $v\in(4\bbZ)^2$ and faces lie on $f\in(4\bbZ+2)^2$. The sites are then labelled by $(2\bbZ+1)^2$. Physically, they correspond to the coordinates of the center of mass of pairs of particles, which we denote 
\begin{equation*}
X = \frac{1}{2}(f+v) \in (2\bbZ+1)^2.
\end{equation*}
Their relative position is given by the complementary coordinate
\begin{equation*}
d = \frac{1}{2}(f-v)\in (2\bbZ+1)^2,
\end{equation*}
and a general two-particle wavefunction is given by $\psi(X,d)$ in $\ell^2((2\bbZ+1)^2\times(2\bbZ+1)^2)$.

\subsection{Dirac anyons}

We now consider a two-particle state having a charge $\chi\in\hat G$ at $v\in\caV$ and a charge $g\in G$ at $f\in\caF$. We let $\lambda^\epsilon = \lambda^\mu = \lambda$ to ensure that the duality symmetry holds, and denote $\lambda^{\epsilon\mu}=\rho$. We shall label the Hamiltonian $K^{\chi,g}_{\lambda,\rho}$. As we shall see, the anyonic nature of the composite particle will reveal itself in the appearance of a non-trivial spectrum. 

Let us first consider the case $\lambda = 0, \rho = 1$. Let $S = \{(1,1),(1,-1),(-1,1),(-1,-1)\}$. The results of Section~(\ref{sec:Dynamics}), see also Figure~\ref{fig:hoppinG phases}, imply that the subspace $\ell^2((2\bbZ+1)^2\times S)$ is invariant for $K^{\chi,g}_{0,1}$, and that its orthogonal complement is the kernel of the Hamiltonian. After Fourier transforming with respect to the variable $X$, we conclude that, for $k\in\bbT^2$, 
\begin{equation*}
        K^{\chi,g}_{0,1}(k) = 2 
            \begin{pmatrix}
0 & \cos(2k_y) & \cos(2k_x) & 0\\
\cos(2k_y) & 0 & 0 & \chi(g) \cos(2k_x) \\
\cos(2k_x) & 0 & 0 &\cos(2k_y)\\
0 & \overline{\chi}(g) \cos(2k_x) & \cos(2k_y) & 0\\
\end{pmatrix}.
\end{equation*}
As discussed earlier, the phase $\chi(g)$ is accumulated by the hopping of the particles across each other's string, which happens only in the case $(1,-1)\leftrightarrow(-1,-1)$, see again Figure~\ref{fig:hoppinG phases}. We obtain four bands given explicitly by
\begin{equation*}
E^{\chi,g}_{0,1}(k) = \pm2\sqrt{\cos^2(2k_x)+\cos^2(2k_y)\pm\vert \cos(2k_x)\cos(2k_y)\vert \sqrt{2+2\mathrm{Re}\chi(g))}}.
\end{equation*}
Generically, this exhibits a Dirac cone at $(\pi/4,\pi/4)$, see Figure~\ref{fig:Dirac cone}. The case $\chi(g) = -1$, which arises for example for $G = \bbZ_2$, is special in that two sheets are exactly degenerate.
\begin{rem}
The case $\chi(g) = 1$ is trivial with dispersion relation $\pm 2\cos(2k_x)\pm2\cos(2k_y)$. It corresponds to simple bosons with no phase being accumulated as the strings cross each other. Unsurprisingly maybe, the anyonic nature of the quasi-particles manifests itself in the spectrum of the two-particles Hamiltonian.
\end{rem}

\begin{figure}
    \centering
    \begin{minipage}{0.49\textwidth}
        \centering
        \includegraphics[width=\textwidth]{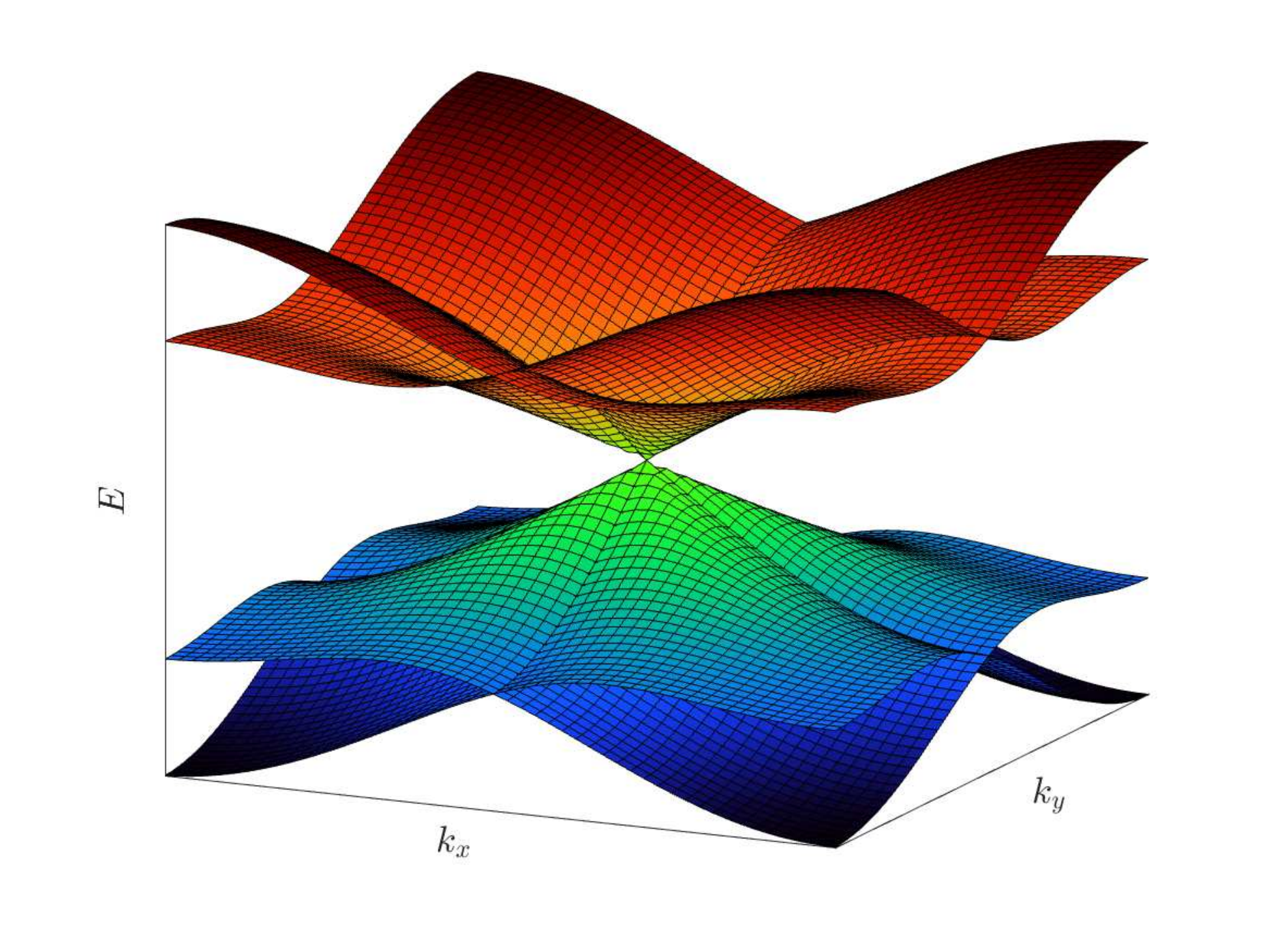} 
        \caption{The Dirac cone appearing at the point $(\frac{\pi}{4},\frac{\pi}{4})$, here for $\chi(g) = \ep{\frac{2\pi\iu}{3}}$.}\label{fig:Dirac cone}
    \end{minipage}\hfill
    \begin{minipage}{0.49\textwidth}
        \centering
        \includegraphics[width=\textwidth]{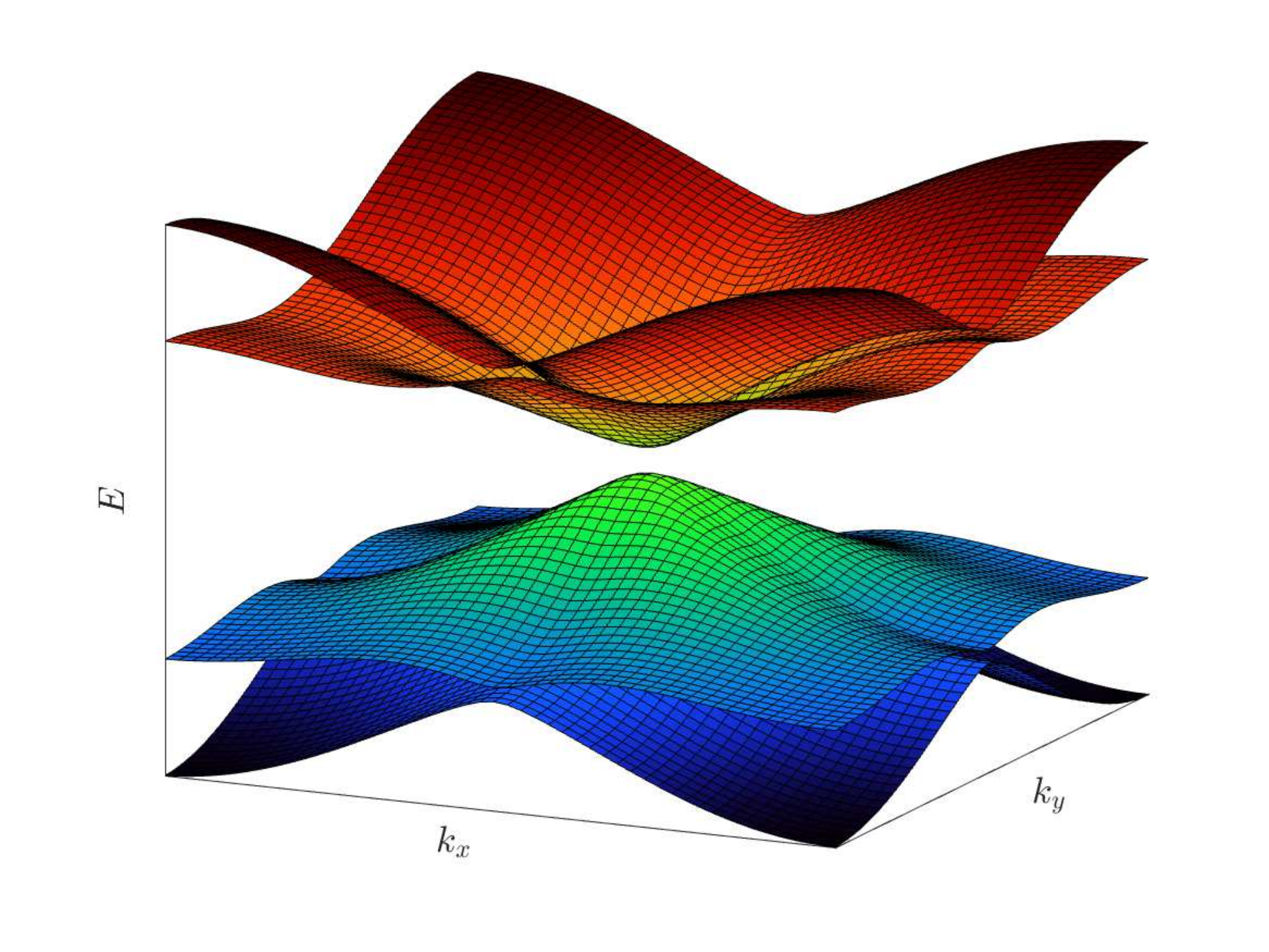} 
        \caption{Gap opening in the presence of symmetry breaking, again for $\chi(g) = \ep{\frac{2\pi\iu}{3}}$.}\label{fig:Avoided cone}
    \end{minipage}
\end{figure}


The Dirac cone is a consequence of the duality symmetry, and it is not stable under the breaking of the symmetry. If we consider an asymmetric hopping by giving a different weight to the two terms of~(\ref{em hopping}), namely $(
\mathcal{F}^{\overline{\chi}}_{(v,e)} \plus{\iota}_{v'} \plus{\chi}_{v} + \mathcal{ F}^{\overline \chi}_{(v',e)} \plus{\iota}_{v} \plus{\chi}_{v'}
)(\caB_f^h + \caB_{f'}^h)$ and $(1+m)(
\mathcal{F}^{\bar{h}}_{(f,e)} \face^{1}_{f'} \face^{h}_{f} + \mathcal{F}^{\bar h}_{(f',e)} \face^{1}_{f} \face^{h}_{f'}
)(\caA_v^\chi + \caA_{v'}^\chi)$, with $m>0$, the dispersion relation then reads
\begin{equation*}
E^{(\chi,g)}_{0,1}(k) = \pm\sqrt{\vert c(k_x)\vert^2+\vert c(k_y)\vert^2\pm\vert c(k_x)\vert\vert c(k_y)\vert\sqrt{2+2\mathrm{Re}\chi(g)}}.
\end{equation*}
where $c(s) = \ep{2\iu s} + (1+m)\ep{-2\iu s}$ opening up a gap of size $2 m$ at the Dirac point, see Figure~\ref{fig:Avoided cone}.

\subsection{Bound states and scattering states.} The dynamics discussed above is that of a strictly bound pair of a flux and a charge. In the presence of the additional terms producing independent hoppings of the charge and the flux, we now show that the bound pair may become unstable and decay into two freely moving particles scattering away from each other.

First of all we prove that, in the extreme case $\rho = 0$ and $\lambda \neq 0$, there is no bound state of the two particles at all. The restriction of $H^\epsilon+H^\mu$ to the invariant space $\caH_{\chi,g}$ is immediately read off~(\ref{eHops},\ref{muHops}). As in Section~\ref{sec: Hamiltonian warm up}, it reduces to a simple discrete Laplacian, up to the phase picked up whenever the $\chi$ particle crosses the string $\overline\gamma$ associated with the $g$ flux and reciprocally. After subtraction of the constant, we obtain
\begin{equation}\label{HL gauge}
K^{\chi,g}_{1,0}(k) = 2\cos({2 k_{x}})\Delta \otimes\idtyty+2 \cos(2k_y) \idtyty \otimes \Delta
+2\cos({2 k_{x}}) \left(\chi(g) \sigma^+ + \overline{\chi(g)}\sigma^- - \sigma^x\right) \otimes \Theta,
\end{equation}
acting on the fibers $\ell^2(2\bbZ+1)\otimes \ell^2(2\bbZ+1)$ and corresponding to the dynamics of the relative coordinate. Here
\begin{equation*}
(\Delta\psi)(z) = \psi(z-2) + \psi(z+2)
\end{equation*}
for any $z\in 2\bbZ + 1$ and $\Theta$ is the multiplication operator by the characteristic function of $\{j\in2\bbZ+1,j<0\}$. In the last term, $\sigma^\sharp$ are the Pauli matrices acting in the space spanned by $|\pm1\rangle$ (as a function in $\ell^2(2\bbZ + 1)$, the vector $\vert m\rangle$ for $m\in2\bbZ + 1$ is the normalized sequence whose only non-zero element is at $m$), and this term amounts to replacing the trivial phase accumulated along a hopping by $\chi(g)$, respectively $\overline\chi(g)$, whenever the two strings cross. 
\begin{prop}\label{No bound state}
Let $k\neq(\frac{\pi}{4},\frac{\pi}{4})\:(\mathrm{mod}\frac{\pi}{2})$. The pure point spectrum of $K^{(\chi,g)}_{1,0}(k)$ is empty.
\end{prop}
\begin{proof}
We recall the setting of~\cite{GrafSchenker}, which is very close to the current situation. For $j=1,2$, let $T^A_j\in\caL(\ell^2((2\bbZ+1)^2))$ be the magnetic translations in the direction of the unit vectors $e_j$, associated with a magnetic vector potential $A: (2\bbZ+1)^2\times(2\bbZ+1)^2\to\bbR$:
\begin{equation*}
(T_{e_j}\psi)(x) = \ep{-\iu A(x,x+2e_j)}\psi(r+2e_j).
\end{equation*}
Here $A$ is so that $-A(x,y) = A(y,x)$. A simple calculation yields that $(T_{e_j})\str = (T_{e_j})^{-1} = T_{-e_j}$. They satisfy the following commutation relations
\begin{equation*}
T_{e_j}^{-1}T_{e_k}^{-1}T_{e_j}T_{e_k} = \ep{-\iu\phi_{(e_k,e_j)}},
\end{equation*}
where
\begin{equation*}
\phi_{(e_k,e_j)}(x) = A(x,x-2e_j) + A(x-2e_j, x-2e_j-2e_k) - A(x-2e_k,x-2e_k-2e_j) - A(x,x-2e_k)
\end{equation*}
is naturally interpreted as the flux through the face centred at $x-e_k-e_j\in(2\bbZ)^2$; In particular it is equal to zero whenever $j=k$. Note that $\ep{-\iu\phi_{(e_k,e_j)}}$ is to be understood as the corresponding multiplication operator. 

The magnetic Hamiltonian is given by
\begin{equation*}
H = c_1(T_{e_1} + T_{-e_1}) + c_2(T_{e_2} + T_{-e_2}),
\end{equation*}
where $c_1,c_2$ are real constants. Let $q(x) = \langle x,x\rangle$, where $\langle,\rangle$ is the restriction of the Euclidean inner product to $(2\bbZ+1)^2$. Let now
\begin{equation*}
A = \iu[H,q]. 
\end{equation*}
We observe that for any function $f$,
\begin{equation*}
[T_{e_j},f] = (D_{e_j}f)T_{e_j}
\end{equation*}
where $(D_{e_j}f)(x) = f(x+2e_j) - f(x)$, and that 
\begin{equation*}
D_{e_j}D_{e_k}q = 8 \langle e_j,e_k\rangle.
\end{equation*}
A computation now yields that
\begin{equation*}
\frac{\iu}{8}[H, A] = -\sum_{j,k=1}^2(T_{e_j}-T_{-e_j}) c_j \langle e_j,e_k\rangle c_k(T_{e_k}-T_{-e_k}) + C,
\end{equation*}
where
\begin{equation*}
C = \sum_{j, k=1}^2(c_jc_k)(D_{e_k}q)\left[T_{e_j} + T_{-e_j},T_{e_k}\right] + (c_jc_k)(D_{-e_k}q)\left[T_{e_j} + T_{-e_j},T_{-e_k}\right]
\end{equation*}
is only non-zero because of the non-commutativity of the magnetic translations, since
\begin{equation*}
[T_{e_j},T_{e_k}] = \left(\idtyty - \ep{-\iu\phi_{(-e_j,-e_k)}})\right)T_{e_j}T_{e_k}.
\end{equation*}

Let us now consider the specific case $K^{\chi,g}_{1,0}(k)$, namely
\begin{equation*}
c_1 = 2\cos(2k_x),\quad c_2 = 2\cos(2k_y),\quad \ep{-\iu A(x,x+2e_j)}=\chi(g) \delta_{j,1}\delta_{x_1,-1}\delta_{x_2<0}.
\end{equation*}
Here, the magnetic vector potential corresponds to a flux tube through the plaquette at $(0,0)$. Hence, the factor $1 - \ep{-\iu\phi_{(\mp e_j,\mp e_k)}(x)}$ appearing in $C$ does not vanish if and only if $x\pm e_j\pm e_k = (0,0)$. In particular, it vanishes if $x\notin S$ and hence $C$ is a finite rank operator.

We now note that $q(x) = \Vert x\Vert^2$ is constant on $S$ and hence $(D_{e_k}q)(x)$ vanishes for all $(e_k,x)$ such that $\phi_{(-e_j,-e_k)}(x)\neq 0$. It follows that $C=0$ and
\begin{equation*}
\frac{\iu}{8}[K^{\chi,g}_{1,0}(k), A] = -\sum_{j=1}^2c_j^2(T_{e_j}-T_{-e_j})^2 = \sum_{j=1}^2c_j^2(T_{e_j}-T_{-e_j})\str(T_{e_j}-T_{-e_j}),
\end{equation*}
which is the sum of two non-negative, but mutually non-commuting, terms. The triangle inequality yields
\begin{equation*}
\vert ((T_{e_j}-T_{-e_j})\psi)(x)\vert\geq \big\vert\vert\psi(x+2e_j)\vert - \vert\psi(x-2e_j)\vert\big\vert
=\big\vert ((T_{e_j}^0-T_{-e_j}^0)\vert\psi\vert)(x) \big\vert,
\end{equation*}
where $T_{e_j}^0$ are translation operators with $A=0$. This is a discrete version of the diamagnetic inequality. Hence,
\begin{equation*}
\frac{1}{8}\langle \psi,\iu[K^{\chi,g}_{1,0}(k), A]\psi\rangle
\geq \sum_{j=1}^2c_j^2\big\Vert (T_{e_j}^0-T_{-e_j}^0)\vert\psi\vert\big\Vert^2
\end{equation*}
and it suffices to show that $\mathrm{Ker}(T_{e_j}^0-T_{-e_j}^0) = \{0\}$ to rule out pure point spectrum. That is immediate since the only square summable solution of
\begin{equation*}
((T_{e_j}^0-T_{-e_j}^0)\psi)(x) = \psi(x+2e_j)-\psi(x-2e_j) = 0
\end{equation*}
is the zero function.
\end{proof}

We now turn to the full Hamiltonian $K^{\chi,g}_{\lambda,\rho}$ in the two-particle sector. Let $P$ be the orthogonal projection onto the space spanned by $\vert \pm1\rangle$. Then
\begin{equation*}
K^{\chi,g}_{0,1}(k) = \cos({2 k_{x}})\sigma^{x} \otimes(\sigma^z+P)+2 \cos(2k_y) P \otimes \sigma^{x}
+\cos(2k_x)\left(\chi(g) \sigma^+ + \overline{\chi(g)} \sigma^-\right) \otimes(P-\sigma^z),
\end{equation*}
which, combined with $K^{\chi,g}_{1,0}(k)$ given in~(\ref{HL gauge}), yields the full Hamiltonian. 
\begin{prop}
For any $k\in\bbT^2$, let $R_\lambda(k) = 4 \lambda \left(\vert \cos(2k_x)\vert + \vert\cos(2k_y)\vert\right)$. Then
\begin{equation*}
\sigma_{\mathrm{ess}}\left(K^{\chi,g}_{\lambda,\rho}(k)\right) = \left[-R_\lambda(k),R_\lambda(k)\right],
\end{equation*}
in particular, the essential spectrum is independent of $\rho$.
\end{prop}
\begin{proof}
Since $K^{\chi,g}_{\lambda,\rho}(k) = K^{\chi,g}_{\lambda,0}(k) + K^{\chi,g}_{0,\rho}(k)$, the independence of the essential spectrum on the parameter~$\rho$ follows from the fact that $K^{\chi,g}_{0,\rho}(k)$ is a finite rank perturbation. It therefore suffices to consider~(\ref{HL gauge}), which is the magnetic operator with a simple flux tube considered in Proposition~\ref{No bound state}. The claim follows from~Section~2 of~\cite{DeNittisSchuba}: The vector potential given here is gauge equivalent to the Aharonov-Bohm gauge $\tilde A$ and therefore $K^{\chi,g}_{\lambda,0}(k)$ is unitarily equivalent to the operator $\tilde K^{\chi,g}_{\lambda,0}(k)$ given by magnetic translations with $\tilde A$ replacing~$A$. If we now denote $K_{\lambda,0}(k)$ the operator corresponding to $A=0$ (it does not depend on $(\chi,g)$ indeed), then $\tilde K^{\chi,g}_{\lambda,0}(k) - K_{\lambda,0}(k)$ is compact by Proposition~3 in~\cite{DeNittisSchuba}. Hence the essential spectrum of the full $K^{\chi,g}_{\lambda,\rho}(k)$ equals that of $K_{\lambda,0}(k)$, which is a sum of two independent discrete Laplacians, concluding the proof.
\end{proof}
We conclude with a proof that $K^{\chi,g}_{\lambda,\rho}(k)$ has bound states for a range of parameters $\lambda,\rho$. If $\lambda>0$ then in particular $K^{\chi,g}_{\lambda,\rho}(k)$ has a bound state for all $\rho>0$, in a neighbourhood of the two lines $(k_x,\frac{\pi}{4})$ and $(\frac{\pi}{4},k_y)$.
\begin{thm}
Let $\left\vert 1+\frac{\rho}{\lambda}\right\vert>1$ and let $k_x\neq\frac{\pi}{4}\:(\mathrm{mod}\frac{\pi}{2})$. Then there is a neighbourhood $\caU$ of $(k_x,\frac{\pi}{4})$ such that if $k\in\caU$, then $K^{\chi,g}_{\lambda,\rho}(k)$ has four possibly degenerate eigenvalues; The corresponding eigenstates are exponentially localized.
\end{thm}
\begin{proof}
First of all,
\begin{align*}
K^{\chi,g}_{\lambda,\rho}\left(k_x,\frac{\pi}{4}\right) = \lambda \cos(2k_x) \Big\{ & \left(\chi(g) \sigma^+ + \overline{\chi(g)}\sigma^-\right)\otimes \left(2\Theta + \frac{\rho}{\lambda}(P-\sigma^z)\right) \\
& + 2\Delta\otimes\mathbb{I} + \sigma^x\otimes\left(-2\Theta+ \frac{\rho}{\lambda}(P+\sigma^z)\right) \Big\}.
\end{align*}
As should be expected, this operator is diagonal in the second tensor factor. Therefore, it suffices to search for eigenstates of the form $\psi\otimes \vert j\rangle$ for any $j\in 2\bbZ+1$. Then $K^{\chi,g}_{\lambda,\rho}(k_x,\frac{\pi}{4}) \psi\otimes \vert j\rangle =  2\lambda \cos(2k_x) (\Delta_j^{\chi,g}\psi) \otimes  \ket{j}$, where the matrix elements of $\Delta_j^{\chi,g}$ are
\begin{equation*}
(\Delta_j^{\chi,g})_{m,n} = \begin{cases}
1 & \text{if }\vert m-n\vert = 2 \text{ and }(\vert m\vert,\vert n\vert )\neq (1,1) \\
b_j & \text{if }(m,n) = (-1,1) \\
\overline{b_j} & \text{if }(m,n) = (1,-1) \\
0 & \text{otherwise}
\end{cases}
\end{equation*}
and the values of $b_j$ are given by
\begin{equation*}
b_j = \begin{cases}
\chi(g) & \text{if }j<-1 \\
\chi(g)(1+ \frac{\rho}{\lambda}) & \text{if }j=-1 \\
(1+ \frac{\rho}{\lambda}) & \text{if }j=1 \\
1 & \text{otherwise}
\end{cases}
\end{equation*}
This is amenable to a transfer matrix analysis. If $\psi$ is an eigenvector of $\Delta_j^{\chi,g}$ for the eigenvalue $E$, then the vector $\phi$ given by $\phi_n = (\psi_{2n+1},\psi_{2n-1})^T$ solves
\begin{equation*}
T\phi_n = \phi_{n+1},\qquad 
T = \begin{pmatrix}
E & -1\\
1 & 0
\end{pmatrix}
\end{equation*}
for all $n\in\bbN\setminus\{-1,0\}$, with boundary condition given by
\begin{equation*}
\phi_0 = 
\begin{pmatrix}
E/b_j & -1/b_j\\
1 & 0
\end{pmatrix}\phi_{-1},\qquad
\phi_1 = 
\begin{pmatrix}
E & -\overline{b_j}\\
1 & 0
\end{pmatrix}\phi_0,
\end{equation*}
namely
\begin{equation}\label{transfer BC}
\phi_1 = B_j\phi_{-1},\qquad 
B_j = \begin{pmatrix}
(E^2-\vert b_j\vert^2)/b_j & -E/b_j\\
E/b_j & -1/b_j
\end{pmatrix}.
\end{equation}
The two eigenvalues of the transfer matrix are equal to $\tau_\pm=\frac{1}{2}(E\pm\sqrt{E^2-4})$ with eigenvectors $v_\pm= (
\tau_\pm,1)^{\mathrm T}$. The eigenvector gives rise to a solution of the eigenvalue equation that decays exponentially at $+\infty$, respectively $-\infty$, if the corresponding eigenvalue has modulus $<1$, respectively $>1$. Both necessarily happen simultaneously since $\tau_+\tau_- = 1$. If $\vert E\vert>2$, then $\tau_\pm$ are real with $\tau_+>1$. Otherwise $\tau_- = \overline{\tau_+}$ and hence $\vert \tau_\pm\vert = 1$. Therefore, a normalizable solution of the eigenvalue equation may only exist if $\vert E\vert>2$. In that case, we set $\phi_1 = \alpha v_-$ for some $\alpha\neq 0$ and $\phi_{-1} = v_+$ and the boundary condition~(\ref{transfer BC}) reduces after elimination of $\alpha$ to an equation for the energy
\begin{equation*}
(E^2 - \vert b_j\vert^2)\tau_+^2 =2 E\tau_+ - 1
\end{equation*}
or equivalently 
\begin{equation*}
\vert b_j\vert^2 = \tau_+^2.
\end{equation*}
Since $\vert \tau_+\vert>1$, we conclude that a solution exists if and only if $\vert b_j\vert>1$, which in turn holds if and only if $j=\pm 1$ and $\left\vert 1+ \frac{\rho}{\lambda}\right\vert >1$. The eigenvalues are given by
\begin{equation*}
E = \tau_+ +\tau_- = \pm\left(B + B^{-1}\right),\qquad B = \left\vert 1+ \frac{\rho}{\lambda}\right\vert,
\end{equation*}
each of them being doubly degenerate since $\vert b_1\vert = \vert b_{-1}\vert$.

If $\left\vert 1+ \frac{\rho}{\lambda}\right\vert >1$, the rest of the spectrum of $\Delta_j^{\chi,g}$ lies in $[-2,2]$ and corresponds to oscillatory solutions. In particular, the spectrum of $K^{\chi,g}_{\lambda,\rho}(k_x,\frac{\pi}{4})$ is characterized by a gap $g(k_x)$ given by
\begin{equation*}
g(k_x) = 2\lambda\left\vert \cos(2k_x) \right\vert \left(B + B^{-1} -2\right),
\end{equation*}
and $g(k_x)>0$ for all $k_x\neq\frac{\pi}{4}\:(\mathrm{mod}\frac{\pi}{2})$. Since the operator-valued function $k\mapsto K^{(\chi,g)}_{\lambda,\rho}(k)\in\caL(\ell^2((2\bbZ+1)^2))$ is smooth, the existence of four branches of eigenstates in a neighbourhood of $(k_x,\frac{\pi}{4})$ away from the degenerate points $k_x\neq\frac{\pi}{4}\:(\mathrm{mod}\frac{\pi}{2})$ is ensured by spectral perturbation theory.
\end{proof}

A similar, and simpler, analysis can be carried out around $(\frac{\pi}{4},k_y)$. In that case, the Hamiltonian reduces to
\begin{equation*}
K^{\chi,g}_{\lambda,\rho}(\frac{\pi}{4},k_y) =2 \lambda\cos(2k_y) \left( \idtyty \otimes \Delta + \frac{\rho}{\lambda} P \otimes \sigma^{x}\right)
\end{equation*}
which is now diagonal with respect to the first factor. It is a simple Laplacian (and therefore has no bound state) on the range of $\idtyty - P$, while 
\begin{equation*}
K^{\chi,g}_{\lambda,\rho}(\frac{\pi}{4},k_y)P\otimes \idtyty = 2 \lambda\cos(2k_y) P\otimes \left(\Delta + \frac{\rho}{\lambda} \sigma^{x}\right)
\end{equation*}
which is of the same form as above, with $b_{\pm1}\to 1+\frac{\rho}{\lambda}$. We conclude again that two bound states exist for any $k_y$ provided $\left \vert 1+\frac{\rho}{\lambda}\right\vert>1$ and they are stable away from the degenerate points.

\begin{figure}
    \centering
    \includegraphics[width=0.6\textwidth]{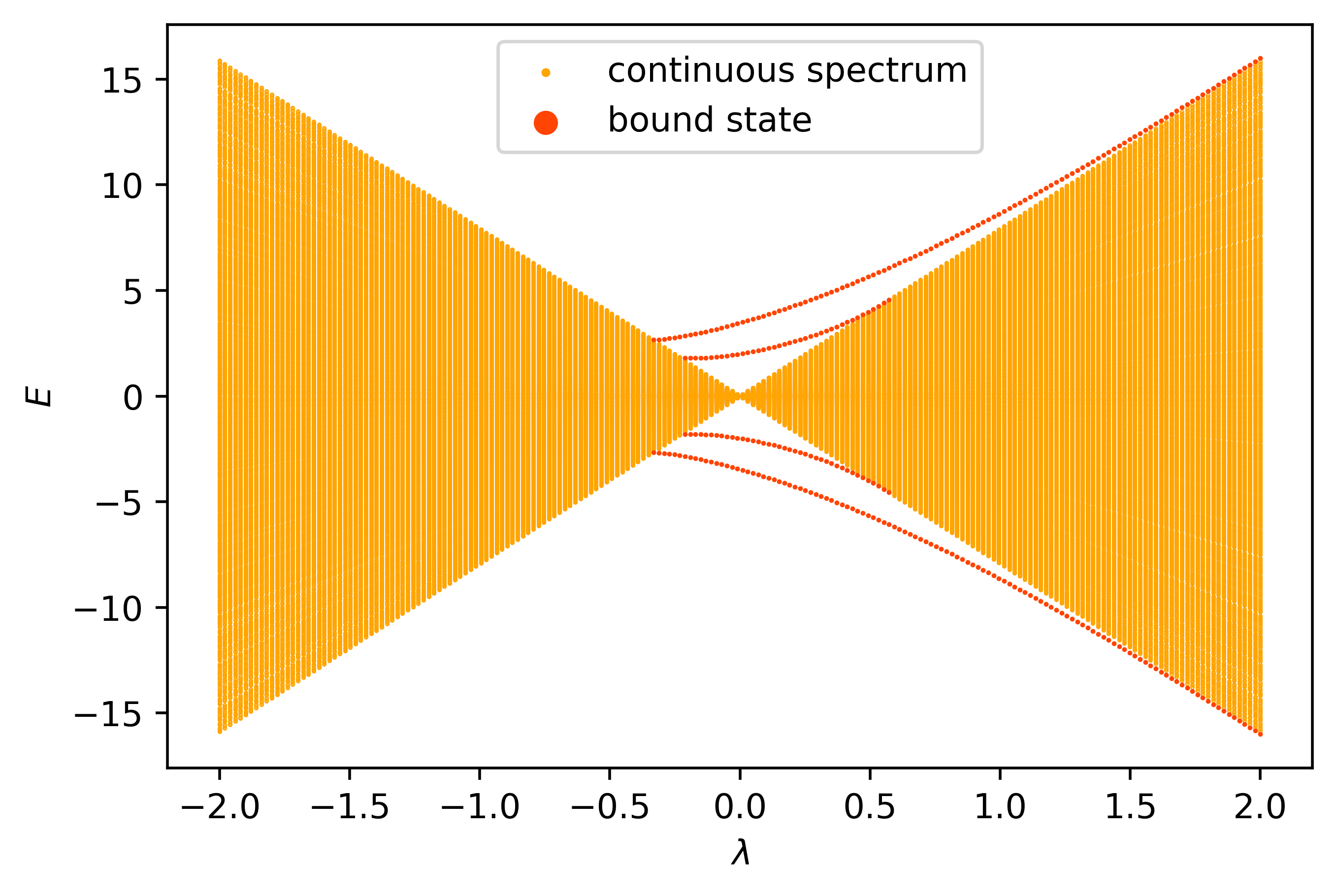}
    \caption{The numerically computed spectrum of $K^{\chi,g}_{\lambda,1}\left(0,0\right)$ as a function of $\lambda$, namely well away from the points where transfer matrix methods yield analytical solutions. The two degenerate branches observed whenever $k_x=\frac{\pi}{4}$ or $k_y=\frac{\pi}{4}$ are now split, clearly showing the four branches of eigenvalues in a wide range of parameters.}
    \label{fig:Kx_Evslam}
\end{figure}

\subsection*{Acknowledgements}
Based upon work supported by the National Science Foundation under grants DMS--1813149 and DMS--21083901 (B.N. \& S.V.). S.B. was supported by NSERC of Canada. This work was initiated while B.N. and S.B. were attending the `Spectral Methods in Mathematical Physics' program held at the Institut Mittag-Leffler, April 2019. BN gratefully acknowledges kind hospitality at the Technical University Munich during the final stages of this work and the Alexander von Humboldt Foundation for support provided through a Carl Friedrich von Siemens Research Award.

\bibliographystyle{unsrt}
\bibliography{RefsDTCM.bib,qss}

\end{document}